%% file: main.tex
\documentclass[10pt, twocolumn, letterpaper]{IEEEtran}

\usepackage{cite}
\usepackage{amsmath,amssymb,amsfonts}
\usepackage{amsthm}
\usepackage{indentfirst}
\usepackage{verbatim}
\usepackage{algorithm}
\usepackage{enumitem}

\usepackage{multicol}
\usepackage[center]{caption}

\newtheorem{lem}{Lemma}

\usepackage{algorithmic}
\usepackage{graphicx}
\usepackage{textcomp}
\usepackage{mathtools}
\usepackage{subfigure}
\usepackage{xcolor}

\usepackage[short]{optidef}

\usepackage{xcolor}
\def\BibTeX{{\rm B\kern-.05em{\sc i\kern-.025em b}\kern-.08em
    T\kern-.1667em\lower.7ex\hbox{E}\kern-.125emX}}

\begin{document}
\setlength{\belowdisplayskip}{2pt}

\title{Location Privacy and Spectrum Efficiency Enhancement in Spectrum Sharing Systems}

\vspace{-.2cm}

\author{\IEEEauthorblockN{Long Jiao, \textit{Member, IEEE,} Yao Ge, \textit{Member, IEEE,} Kai Zeng, \textit{Member, IEEE,} B.C. Hilburn}
\thanks{This work is mainly conducted at Wireless Cyber Center, George Mason University, and partially conducted at the University of Massachusetts Dartmouth. This work was supported in part by the Microsoft Research Award, US Army Research Office (ARO) through grant No. W911NF-21-1-0187, the Commonwealth Cyber Initiative (CCI) and its Northern Virginia (NOVA) Node, an investment in the advancement of cyber R\&D, innovation and workforce development. (Corresponding author: Kai Zeng.)}
\thanks{Long Jiao is with the Department of Computer \& Information Science, University of Massachusetts Dartmouth, North Dartmouth, MA 02747 USA (email: ljiao@umassd.edu).}

\thanks{Yao Ge is with the Continental-NTU Corporate Lab, Nanyang Technological University, Singapore (email: yao.ge@ntu.edu.sg).} 

\thanks{Kai Zeng is with the Department of Electrical and Computer Engineering, George Mason University, Fairfax, VA 22030 USA (email: kzeng2@gmu.edu).}

\thanks{B.C. Hilburn is with the Microsoft Corp., USA (email: benjamin.hilburn@microsoft.com).}
}

\maketitle

\begin{abstract}
In this work, we investigate the benefits of secondary user (SU) network beamforming on improving primary user (PU) location privacy in spectrum sharing systems, where the beamformer in the SU network is designed to suppress the aggregate interference to improve the location privacy of PUs. 
We consider two problems: improving SU network communication throughput subject to the specified PU location privacy requirements, and enhancing PU location privacy given the quality of service (QoS) requirements of SU networks. 
In the first problem, we provide an algorithm to achieve high data rates with the constrained PU location privacy level. {Numerical results show that for a given PU location privacy requirement, the proposed scheme is able to interfere/exclude only a few SU nodes from the PU band and the network throughput can be greatly improved.} In the second problem, to fully explore the potentials of SU network beamforming for enhancing PU location privacy, we propose a two-step scheme to decouple the beamforming and privacy zone design so that the PU location privacy can be improved while satisfying the SU network throughput requirement. According to numerical evaluations, the proposed scheme can maintain/achieve higher PU location privacy than the benchmark beamforming schemes while satisfying a QoS requirement for the SU network.

\end{abstract}

\begin{IEEEkeywords}
\vspace{-.1cm}
5G, Spectrum sharing systems, Beamforming, Location privacy, Operational security
\end{IEEEkeywords}
\IEEEpeerreviewmaketitle
\input{Introduction}
\input{System_model}

\input{DataRate_Maximization}

\input{maximizeOP}
\input{Cloak_Region_Design}

\input{Beamformer_design}

\input{Numerical_results}

\input{Appendix2}

\vspace{-.4cm}
\small
\bibliographystyle{IEEEtran}
\vspace{4mm}


\end{document}

%% file: Introduction.tex
\vspace{-.5cm}
\section{Introduction}
\vspace{-.2cm}
Enabled by key technologies such as heterogeneous networks and beamforming, the fifth generation (5G) mobile networks aim to achieve 1,000 times throughput increase and 10 times spectrum efficiency improvement compared with existing networks \cite{agiwal2016next,ahmed2018survey,akpakwu2017survey,jiao2019physical}. For instance, due to the directional-communication ability, beamforming has the potential to realize high spectral efficiency and support a wide range of applications \cite{jang2016smart,mismar2019deep,razavizadeh2014three}, while meeting high throughput requirements. 

Along with the prosperity of 5G and beyond networks, radio spectrum becomes increasingly congested due to the need to support emerging cyber-physical system (CPS) and Internet of Things (IoT) applications, such as virtual/augmented reality, connected vehicles, smart cities, smart manufacturing, etc.
Dynamic spectrum sharing (DSS) is considered as a key technology to improve the spectrum utilization by allowing secondary users (SUs) (e.g., cellular network users) to opportunistically access to the spectrum owned by primary users (PUs), such as military radars \cite{ahmad20205g,sangdeh2019practical,ge2022energy}. 
To enable DSS to achieve high spectrum and energy efficiency, the spectrum access system (SAS)-assisted DSS, a new paradigm, has been proposed to provide interfaces for PUs and SUs for the efficient spectrum sharing in the 3550-3700 MHz band \cite{sohul2015spectrum}. The SAS is responsible to manage the SUs and ensure the priority of spectrum usage for active PUs \cite{sss2013NIIT}.

Despite the great advantages achieved regarding spectrum sharing systems, the investigation of PU privacy protection is still in its infancy. In the operation of SAS-assisted spectrum sharing systems, PUs need to register and exchange sensitive operational information with SAS. Operational information of PUs usually includes geo-location of PUs, PUs' carrier frequency, and time/duration of transmissions, which could reveal PUs' confidential information. For instance, the historical location information for a military PU may reveal confidential information including its identity, its trajectory, and its mission \cite{salama2019trading}. The protection of PUs' operational information is identified as operational security (OPSEC) of PUs \cite{bahrak2014protecting,bhattarai2017thwarting,clark2017trading}. 
OPSEC of PUs is usually subject to two types of inference attacks. 
First, the sensitive information of PUs can be inferred based on registered records at the SAS if adversaries hack into the database of the SAS, or eavesdrop over the communication between the SAS and PUs. This attack is usually referred to as the \emph{SAS-hijack inference attack}. Second, in addition to the spectrum database hijacking, the OPSEC of PUs could be threatened by the \emph{colluding inference attack}. For instance, the adversary can control a group of colluding SUs and make a large number of queries to the SAS database. The accumulated database queries can be gathered to infer PUs' location information \cite{clark2017trading,clark2020optimizing,salama2019trading}.

To protect the PUs' location privacy from the above attacks, the privacy zone has been proposed in the existing works \cite{salamal2020privacy,zhang2015optimal}. {The privacy-zone based solutions have multiple advantages. First, the privacy zone is a PU-specified region locating in the exclusion zone (EZ). For instance, privacy zone locates inside the EZ and thus overlaps a part of EZ. PUs' real location doesn't need to be released. Second, the aggregate interference at privacy zone is restricted so that the PU can function normally in any location of the privacy zone. For instance, the privacy zone can be divided into a group of cells and PUs can locate in any cell because each cell is under a unified aggregate interference threshold.} Each cell in the privacy zone thus has the same probability to be the real PU locations. 

Improving the PU location privacy and enhancing the performance of SU communications based on the privacy zone, however, are two contradictory objectives due to the existence of the aggregate interference: the performance of one objective is directly limited by the other. Let us consider the SUs at/near the boundary of the EZ.  
On one hand, the SU network throughput can be constrained  by a given privacy zone.
To fulfill the aggregate interference threshold imposed by the privacy zone, the transmission power of each SU locating outside of the EZ has to be controlled and some SU transmissions in the shared band even need to be excluded in order to meet the aggregate interference requirement. 
This effect results in the low SU network throughput on the shared band.  
On the other hand, the size of privacy zone (or the privacy level of PUs) will be restricted by the aggregate interference due to the SU communications. With severe aggregate interference from SU networks, it could be possible that only a small area in EZ can be set as the privacy zone, which satisfies the interference threshold requirement. That is, a higher transmission power in SU networks implies high throughput, but at the same time leads to the higher interference to PUs and limits the size of the privacy zone. Therefore, a solution being beneficial to the two objectives is highly desired, but remains elusive.



The underlying reason of the aforementioned dilemma is attributed to the severe SU communication interference. This effect is overlooked in the existing works. Multiple-antenna beamforming, as a trending technique in 5G and NextG wireless networks \cite{bertizzolo2020cobeam}, is particularly efficient for suppressing interference. The beamforming-enhanced PU location privacy protection is promising yet is missing in the existing works. For instance, beamforming has the potential to strike a good tradeoff between the two contradictory objectives so that one objective can be improved without diminishing the other. {First, beamforming can be leveraged to enhance the SU network communication performance given the aggregate interference threshold imposed by the privacy zone. 
The SAS can coordinate SU beamforming to simultaneously increase the data rates and turn the mainlobe of the beam away from the privacy zone.} 
In this case, the interference mitigation ability enabled by beamforming allows SAS to configure the SU communications with much higher transmission power or incorporate more SU nodes in current networks without violating the aggregate interference threshold. 
Secondly, beamforming-based scheme will enable a higher PU location privacy level compared with the ones without beamforming or with the standard beamformers. 
The SAS can coordinate the SU beamformings to reduce the aggregate interference from SUs so that the size of the privacy zone can be expanded while meeting the SU network throughput requirement.
As a result, the uncertainty of real PU locations is increased after expanding the privacy zone and the location privacy of PUs is thus enhanced.
However, the benefit of beamforming on the PU location privacy protection has not been explored. To fill this research gap and embrace the opportunity brought by beamforming, for the first time, we aim to improve the PU location privacy by leveraging the SU network beamforming.

To fully explore the benefits of beamforming on improving the PU location privacy, we formulate two problems.
In the first problem, we aim to improve SU network communication throughput given a privacy zone. 
With the interference mitigation ability, beamforming enables the SU networks to serve more SU nodes and provide higher data rates without violating the aggregated interference threshold for the given privacy zone. 
In the second problem, we investigate the benefits of beamforming on improving the PU location privacy given the SU network throughput requirement.
The proposed scheme can directly suppress the aggregate interference via beamformer design so that the aggregated interference in a large area within EZ can be controlled under a threshold. 
Therefore, beamforming enables SAS to configure a much larger privacy zone while satisfying the SU network throughput requirement.

To sum up, the major contributions of this work are as follows.
\begin{itemize}[leftmargin=0.4cm]
    \item For the first time, we investigate the potentials of beamforming techniques on enhancing PU location privacy. We utilize beamforming to strike a good balance between two contradictory objectives, i.e., improving SU network data rates and enhancing PU location privacy. {Problem 1 focuses on the location-privacy-first case and PU has the priority to the spectrum usage. Problem 2 enables to balance the PU location privacy and SU network throughput, which can increase the efficiency of spectrum sharing and thus increase revenue.}
    We formulate two problems.
    
    \item In the first problem, the beamforming techniques are applied to improve SU network communication throughput with a given privacy zone (with the specified PU location privacy). Compared with the schemes with standard beamformers, the proposed scheme can serve more SU nodes with higher network throughput.
    
    \item In the second problem, we aim to improve PU location privacy performance using beamforming techniques while satisfying the SU network throughput requirement. Beamforming changes rapidly along with the fast fading wireless channels, whereas the privacy zone holds static for a long time period. Due to this reason, we develop a long-term upper bound for the antenna gain and our design relies on long-term network deployment information. Therefore, the second problem is not the first problem's reciprocal. {Numerical results show that for a given privacy zone, the proposed scheme has a low cutoff threshold (-110 dBm) with the SU node exclusion, while the traditional maximum ratio combining (MRC) beamformer-based scheme couldn’t maintain the same privacy level after  -104 dBm).}
    
\end{itemize}

The rest of this paper is organized as follows. Section \ref{R_works} introduces the existing works. Section \ref{Sys_model} provides the system architecture and basic concepts of SAS-assisted spectrum sharing systems, and privacy metrics.
In section \ref{sec_DataRateMax}, we formulate a problem by utilizing the SU beamforming to maximize the data rate under the constraints of PU location privacy, and an efficient solution is given. 
In section \ref{Prob_formulation}, we focus on the second scenario to improve PU location privacy while satisfying the SU network throughput requirement.
We propose a two-step scheme to decouple the inter-dependency between beamformer design and privacy zone design in the optimization problem.
Section \ref{POIUIST} provides simulation results and Section \ref{Con_cl} concludes the paper. Some detailed proofs appear in the Appendix.

\vspace{-.2cm}
\section{Related work}
\label{R_works}
\vspace{-.1cm}

{Location privacy protection in modern networks is receiving more attentions. To investigate the 3D geolocation protection, the impact factors of the degree of indistinguishability of 3D geolocations is studied in \cite{min20213d}. To achieve the tradeoff between the privacy-preserving level and user utility, a bilateral privacy preservation framework is introduced in \cite{liu2021bilateral}, where a game-theoretic approach based on the Stackelberg model is proposed.} {To protect the sensitive semantic location privacy in VANETs, a reinforcement learning (RL) based scheme is proposed in \cite{min2021reinforcement}. Differential privacy is utilized in this scheme to randomize the released vehicle locations and adaptively selects the perturbation policy.} To protect PUs' location privacy in spectrum sharing systems, many pioneering works have been proposed in the existing literature. 
For example, in \cite{clark2016can}, the authors discussed several attack models and PUs’ obfuscation
strategies. The obfuscation strategies, conducted in the database, include false PU entries insertion and PU parameter noise injection.
In \cite{bahrak2014protecting,zhang2015optimal}, k-anonymity based approaches have been proposed to generate the privacy zone, where k PUs are cloaked in a larger protected contours (privacy zone). 
When PUs are far away from each other, the proposed approach can lead to over-sized privacy zone, which severely reduces the spectrum utilization.
The k-clustering approach is then proposed to divide PUs into k clusters based on their distance. 
To counter the disadvantage of k-anonymity, the authors in \cite{rajkarnikar2017location} proposed to apply the pareto approach to maximize the number of dummy PUs while making EZs as smaller as possible. {Various metrics have been proposed to evaluate the PU location privacy. For instance, in \cite{clark2016can}, the metric is developed to evaluate the expected parameter estimation error. In this case, privacy can be measured directly from this adversary's estimated distribution. In \cite{clark2017trading},  one metric is developed to evaluate the size of the search area for the attackers. The search area consists of the search radius and region area. This search radius reflects the minimum distance around each estimated location the adversary would have to search in order to intercept the actual PU locations. In \cite{bahrak2014protecting}, the authors propose three metrics that focus on different aspects: uncertainty, inaccuracy, and incorrectness. The gap between the SAS’s knowledge and adversarial’s knowledge is evaluated.} What's more, several cryptographic techniques were proposed in \cite{chen2017towards,li2018preserving,dou2017preserving} to design the secure protocol for database-driven spectrum sharing systems. Relying on the techniques of secure computation, the spectrum allocation process is performed over ciphertext
based on homomorphic encryption. 

Unfortunately, the aforementioned works do not consider the interference from SU to PUs and {the trade-off between user privacy and spectrum utility is not covered}.
In \cite{clark2018achievable,clark2020optimizing,salamal2020privacy}, the authors assess this trade-off with both sensing and interface obfuscation approaches in a spectrum sharing system. For instance, in \cite{clark2020optimizing}, the authors formulate the multi-utility user privacy optimization problem. 
The PUs' privacy is measured by exposure to potential adversary inference
attacks. In \cite{salama2019trading}, to guarantee that the interference level remains under required limits for PUs, the authors propose a three-way trade-off among privacy, interference, and utility to systematically study its performance under various conditions. A concept of privacy zone has been proposed to formalize the relationship between privacy and spectrum efficiency. 

The aforementioned existing works have not explored interference mitigation ability of beamforming in SU networks.
This motivates us to focus on enhancing the location privacy of PUs with the beamforming design in SU networks.
We would like to emphasize that our work is substantially different from the objective of throughput maximization using conventional beamformings  \cite{huang2011robust,nasir2017secure}. 
We propose a new beamforming design scheme to mitigate the aggregate interference to PUs in order to strikes a balance between PU location privacy and the SU networks spectrum efficiency.

%% file: System_model.tex
\vspace{-.2cm}
\section{System Model}
\label{Sys_model}
\vspace{-.1cm}

In this section, we begin by introducing the SAS-assisted spectrum sharing systems, specifying its architecture and major components.
We then review the privacy metrics utilized to design the privacy zone. We also discuss the beamforming techniques and the expression of the received signals. 

\vspace{-.3cm}
\subsection{System Architecture}
\vspace{-.1cm}

In this work, we consider a SAS-assisted spectrum access system, which includes the SAS, PUs, SU networks, and environment sensing capability (ESC) \cite{palola2017field} as shown in Fig. \ref{fig:Architecture}. 
SAS is a centralized database, which is responsible for authenticating SUs and managing interference at PUs  based on the registered information of PUs and SUs. PUs originally own the licensed spectrum, which can be shared with SUs, and therefore have priorities to access the spectrum. We consider SU networks that contain secondary receivers (SRs) and multi-antenna secondary transmitter (STs). To detect the presence of PUs, ESC is the sensor network providing the sensing capability.

In order to ensure the priority of active PUs for the spectrum access, SUs must access the spectrum in a manner that the interference from SUs to PUs is restricted. Hence, SUs that are geographically close to the active PUs have to be excluded from transmission when PUs are in operation. For instance, the SAS establishes an exclusion zone (EZ), indicated by the rectangle with red dashed lines in Fig. \ref{fig:System_model}, inside which no SU transmissions are allowed. 
Only the SUs outside the EZ are permitted to operate on the PU's band. EZ is not the only option to exclude SU nodes. We have to emphasis that for the schemes performing the SU exclusion without EZ, privacy zone also can be established and the proposed scheme also applies. 
We consider $M$ cells along the edge of EZ, which will introduce non-negligible interference to PUs. 
In each cell, there is one secondary transmitter (ST), e.g. base station (BS), and multiple secondary receivers (SRs), e.g., user equipment (UE).

\vspace{-.5cm}
\begin{figure}[htbp]
    \centering
    \subfigure[SAS-assisted Spectrum Sharing Systems]{ 
    \includegraphics[width=0.7\linewidth,height=0.5\linewidth]{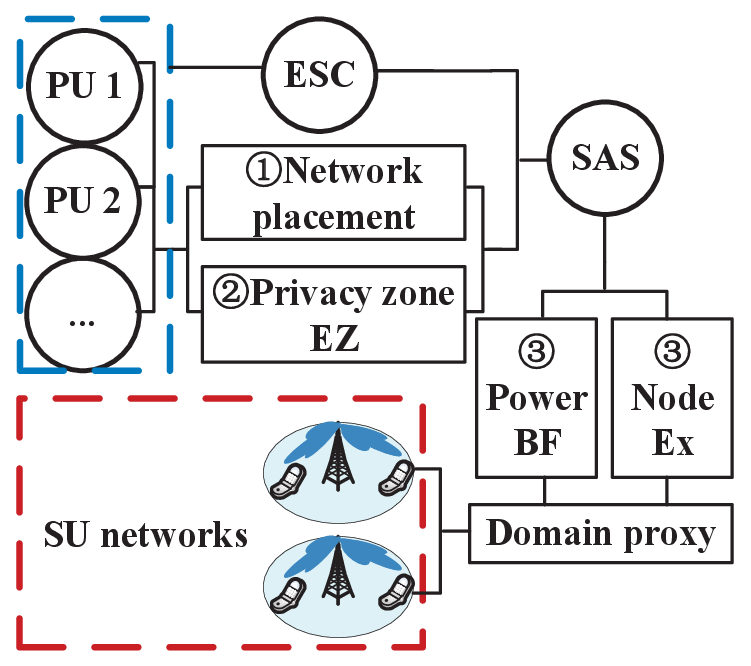}
    \label{fig:Architecture}
    }
    \subfigure[EZ and Privacy Zone.]{  \includegraphics[width=0.7\linewidth,height=0.5\linewidth]{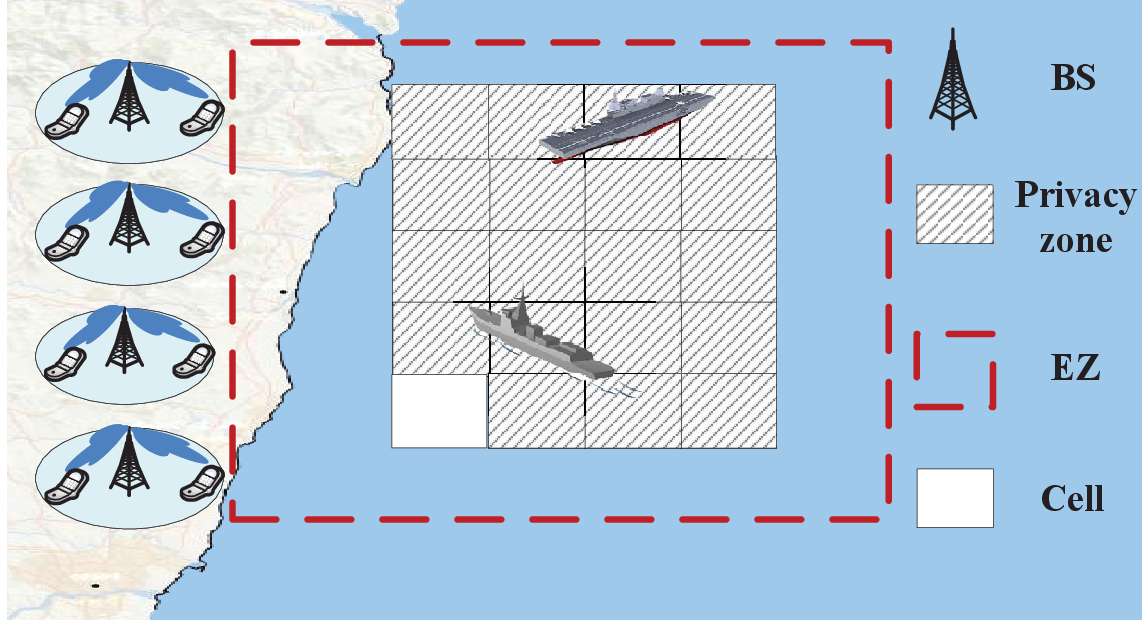}
    \label{fig:System_model}
    }
    \caption{System Model.}
    \vspace*{-.5cm}
\end{figure}

\vspace{-.2cm}
\subsection{Privacy Zone and Privacy Metrics}
\label{UOWETS}

As illustrated in in Fig. \ref{fig:System_model}, a privacy zone, containing multiple grey square cells, has been established \cite{salama2019trading,bahrak2014protecting,zhang2015optimal}. Privacy zone (denoted by the combination of grey square cells) locates inside of the EZ (bounded by the red dash line). The location privacy of PUs hiding in privacy zone is protected by enabling PUs to locate at any grey cells in the privacy zone. For instance, the aggregate interference from STs at each cell in the privacy zone is subject to the same threshold. 
Therefore, from an inference attacker's point of view, the PU has the equal probability to reside in any cell in the privacy zone.

To quantify the PU's location privacy, there are a variety of metrics proposed in the existing literature \cite{clark2017trading,clark2018achievable,salama2019trading}. Without loss of generality, we choose the entropy-based privacy metric, which widely adopted in previous works \cite{salama2019trading,bahrak2014protecting}. Note that our methodology can also be generalized to other privacy metrics in a straightforward manner. 

For a given privacy zone $\mathcal{C}^{\text{PZ}}$, we denote one cell in $\mathcal{C}^{\text{PZ}}$ as $c_{i,j}$. 
Here $(i,j)$ denotes the coordinate of $c_{i,j}$ in $\mathcal{C}^{\text{PZ}}$. 
According to the definition of privacy zone, the probability of PU residing in any $c_{i,j}\in\mathcal{C}^{\text{PZ}}$ is equal from the attacker's perspective. 
Therefore, the probability that one cell $c_{i,j}\in \mathcal{C}^{\text{PZ}}$ is the real location of PUs can be given as $\text{Pr}(\text{PU in~} c_{i,j})=\frac{1}{|{\mathcal{C}^{\text{PZ}}}|}$.
Here $|{\mathcal{C}^{\text{PZ}}}|$ denotes the size of the privacy zone. 
Given the probability above, the entropy of PU location (in bits) can serve as the privacy metric to evaluate the uncertainty level of PU locations, which can be expressed as:
\vspace{-.2cm}
\begin{eqnarray}
H(\mathcal{C}^{\text{PZ}})&=&-\sum_{c_{i,j}\in{\mathcal{C}^{\text{PZ}}}}\text{Pr}(\text{PU in~}{c_{i,j}})\log\Big(\text{Pr}(\text{PU in~}{c_{i,j}})\Big)\nonumber\\
&=& \log(|\mathcal{C}^{\text{PZ}}|).
\label{SSSS_9}
\vspace{-0.5cm}
\end{eqnarray}

Please note, from Eq.~(\ref{SSSS_9}), we can observe that the PU location privacy is logarithmically proportional to the area of the privacy zone because the logarithm function in Eq.~(\ref{SSSS_9}) is monotonically increasing as $|\mathcal{C}^{\text{PZ}}|$ does, where $|\mathcal{C}^{\text{PZ}}|$ stands for the number of cells contained in the privacy zone or the size of $\mathcal{C}^{\text{PZ}}$. Therefore, a larger privacy zone (or a larger $|\mathcal{C}^{\text{PZ}}|$) can increase the uncertainty of PUs' location and have the resilience against inference attacker.

\vspace{-.5cm}
\subsection{Beamforming}
\vspace{-.2cm}

Beamforming can be utilized to obtain sufficient signal-to-noise-ratio (SNR) by concentrating the signal to a certain direction, thus compensate path-loss and save transmission power. 
We assume phased antenna array is used for beamforming at the ST. 
By tuning the phase of the signals on antenna elements in the antenna array, the transmitted signals can be targeted at a specific angular domain.

Let $\mathcal{M} = \{1, 2,...,M\}$ and $\mathcal{N}_m = \{1, 2,...,N_m\}$ ($m\in\mathcal{M}$) represent
the ST set and the corresponding SR set, respectively. Let $\mathcal{P} = \{1, 2,...,P\}$ represents the set indexing $P$ active PUs, which is registered at the SAS. Each multi-antenna ST is assumed to be equipped with an uniform linear array (ULA) antenna with $k$ antenna elements to serve $\mathcal{N}_m$ single-antenna SRs. We utilize a pair $(m,n)$ to denote the transmission link from ST $m$ to SR $n$ ($n\in\mathcal{N}_m$). Let $\tilde{\bf{h}}_{m,n}\in \mathbb{C}^{k\times 1}$ denote the channel coefficients for the link $(m,n)$. 
${\bf{w}}_{m,n}\in \mathbb{C}^{k\times 1}$ denotes the beamformer applied at ST $m$ for the link $(m, n)$.
Considering a multi-user MIMO (MU-MIMO) communication scenario, the received signal at SR n consists of the desired signal coming from its
associated ST $m$, the intra-cell interference from other links $(m, n_1)$ (where $n_1\neq n$ and $n_1\in\mathcal{N}_m$), and the inter-cell interference from adjacent ST $m_1$, which can be expressed as:
\begin{flalign}
y_{m,n}=\tilde{\bf{h}}_{m,n}^H{\bf{w}}_{m,n}s_{m,n}+\underbrace{\sum_{n_1\neq n,~n_1\in\mathcal{N}_m}\tilde{\bf{h}}_{m,n}^H{\bf{w}}_{m,n_1}s_{m,n_1}}_{I_{m,n}^{\text{intra}}~(\text{intra-cell interference})}\nonumber\\
+\underbrace{\sum_{n_1\in\mathcal{N}_{m_1}~,m_1\neq m,~m_1\in\mathcal{M}}\tilde{\bf{h}}_{m_1,n}^H{\bf{w}}_{m_1,n_1}s_{m_1,n_1}}_{I_{m,n}^{\text{inter}}~(\text{inter-cell interference)}}+z_{m,n},
\label{PSOEIOTE}
\end{flalign}
where $s_{m,n}\in \mathbb{C}$ is the unit-power transmitted symbol for the link $(m,n)$ such that  $\mathbb{E}[|s_{m,n}|^2]=1$.
The power budget is given as $\sum_{n\in\mathcal{N}_{m}}u_{m,n}||{\bf{w}}_{m,n}||^2\leq P_m$. The binary variable ${{u}}_{m,n}$ indicates the on/off status of SU link $(m,n)$, i.e., ${u}_{m,n}=1$ indicates the link $(m,n)$ is on/activated on the PU spectrum, and vice versa. $z_{m,n}$ denotes the additive circular symmetric complex Gaussian noise with zero mean and variance $\sigma_{m,n}^2$.
\vspace{-.2cm}
\subsection{Aggregate Interference}

Denote the interference threshold for each PU in $\mathcal{P} = \{1, 2,...,P\}$ as $\{I_{\text{th},1}, I_{\text{th},2},...,I_{\text{th},P}\}$. At first, for each cell in the privacy zone, its aggregate interference threshold cannot reveal the PU information and a unified interference $I_{\text{th}}=\min \{I_{\text{th},1}, I_{\text{th},2},...,I_{\text{th},P}\}$ is thus specified. Second, $I_{\text{th}}$ is applied to every cell in the privacy zone, which requires the aggregate interference at every cell in $\mathcal{C}^{\text{PZ}}$ must be lower than or equal to the threshold $I_{\text{th}}$. For a group of PUs, $I_{\text{th}}$ is set to be the minimum interference threshold among PUs. We denote the aggregate inference at cell $c_{i,j}$ as $f_1^{i,j}({\bf{U}},{\bf{W}})$ under a SU network configuration tuple, i.e., $({\bf{U}},{\bf{W}})$, where ${\bf{U}}=\{u_{m,n} ~|~ u_{m,n}=0~or~1,~ m\in\mathcal{M},n\in\mathcal{N}_m\}$ and ${\bf{W}}=\{{\bf{w}}_{m,n} ~|~m\in\mathcal{M},n\in\mathcal{N}_m\}$. We can express the aggregate interference in dBm \cite{sss2013NIIT,locke2010assessment} as (\ref{Interference_threshold}).
\vspace{-.1cm}
\begin{flalign}
f_1^{i,j}({\bf{U}},{\bf{W}}) =10~\text{log}_{10}(\sum_{m,n}\tilde{I}_{m,n}^{i,j}u_{m,n}||{\bf{w}}_{m,n}||^2 G(\theta_{m,n}^{i,j})),\label{f1_interference}\\ 
f_1^{i,j}({\bf{U}},{\bf{W}})+30\leq I_{\text{th}}, \label{Interference_threshold}
\vspace{-.1cm}
\end{flalign}
{$\tilde{I}_{m,n}^{i,j}$ contains the interference terms such as the insertion loss and propagation loss at cell $c_{i,j}$ incurred by the transmission from ST $m$ to its corresponding SR $n$ \cite{locke2010assessment}. Because the signal traveling distance is usually around tens of kilometers, each point inside the small cell $c_{i,j}$ experiences the similar large scale fading and we thus approximate its the propagation loss by using the coordinates of the centre point in $c_{i,j}$.} Please find the detailed definition of $\tilde{I}_{m,n}^{i,j}$ in \cite{locke2010assessment}. Here $G(\theta_{m,n}^{i,j})$ is the antenna gain determined by the beamforming design for antenna arrays. Its detailed expression is given in Appendix I.

In the following two sections, we will investigate the following two problems: (i) improving SU network throughput based on beamforming techniques where the PU location privacy requirement serves as a constraint;
(ii) enhancing PU location privacy via the privacy zone design given the SU throughput constriants. For each problem, we develop the efficient algorithm as the solution.

%% file: DataRate_Maximization.tex
\vspace{-.2cm}
\section{SU Network Throughput Maximization with PU Location Privacy Constraint}
\label{sec_DataRateMax}

In this section, we will consider the first scenario to maximize the sum data rate under the specified PU location privacy constraint, i.e. given $\mathcal{C}^{\text{PZ}}$. By jointly optimizing the SU node selection (i.e. ${\bf{U}}$) and beamformer design (i.e. ${\bf{W}}$), the optimization problem is formulated as
\vspace{-.05cm}
\begin{subequations}
\label{Prob3_all}
\begin{align}
\mathop {\max }\limits_{{\bf{U}},{\bf{W}}}\quad & \sum_{m\in \mathcal{M},~n\in \mathcal{N}_m} u_{m,n} \log_2\Big(1+\frac{|[\tilde{\bf{h}}_{m,n}^t]^H{\bf{w}}_{m,n}|^2}{I^{\text{inter}}_{m,n}+I^{\text{intra}}_{m,n}+\sigma_{m,n}^2}\Big)
\label{Prob3_obj}\\
\text{s.t.}\quad 
& \big(f_1^{~i,j}({\bf{U}},{\bf{W}})+30\big)\leq I_{th},~\forall c_{i,j}\in\mathcal{C}_{r}^{\text{PZ}},\label{Prob3_con2}
\\
& \log_2\Big(1+\frac{|[\tilde{\bf{h}}_{m,n}^t]^H{\bf{w}}_{m,n}|^2}{I^{\text{inter}}_{m,n}+I^{\text{intra}}_{m,n}+\sigma_{m,n}^2}\Big)\geq u_{m,n}r_{m,n},~\nonumber\\
&\forall m\in \mathcal{M},~n\in \mathcal{N}_m,\label{Prob3_con2_2}\\
& \sum_{n\in\mathcal{N}_m}u_{m,n}||{\bf{w}}_{m,n}||^2\leq P_{m},~\forall m\in\mathcal{M},\label{Prob3_con3_1}\\
&u_{m,n}\in\{0,1\},~\forall m\in\mathcal{M},~n\in \mathcal{N}_{m}.
\label{Prob3_con5}
\end{align}
\end{subequations}
Objective (\ref{Prob3_obj}) aims to maximize the SU network sum data rate, which is the summation of all the active links. Constraint (\ref{Prob3_con2}) sets a limit on the aggregate interference at each cell $c_{i,j}$ in the given $\mathcal{C}_{r}^{PZ}$.
Constraint (\ref{Prob3_con2_2}) puts a minimum data rate $r_{m,n}$ on an active link $(m,n)$ given the CSI estimation $\tilde{\bf{h}}_{m,n}^t$ at time slot $t$.
If there is no feasible solution satisfying this QoS constraint, SU link $(m,n)$ must be deactivated on the PU band (i.e.,  by setting $u_{m,n}=0$). In (\ref{Prob3_con3_1}), the power of each ST is upper bounded. A binary variable $u_{m,n}$ is introduced in (\ref{Prob3_con5}) to indicate the on/off status of link $(m,n)$. 

The problem (\ref{Prob3_all}) is a mixed combinatorial non-convex
optimization problem. In specific, the binary constraints in (\ref{Prob3_con5}) result in the
combinatorial effect while the non-convexity arises in the QoS constraint (\ref{Prob3_con2_2}). 
In general, there is no computational efficient approach to solve (\ref{Prob3_all}) optimally. 
In \cite{wei2017optimal,ng2016power}, the non-convex multi-linear production term has been addressed. However, the techniques developed in \cite{wei2017optimal,ng2016power} cannot be directly applied to our problem since they did not consider the beamformer design. 
In subsequent subsections, we will develop  efficient transformation and approximation techniques to handle such non-convexity and give the corresponding convex/concave expressions. Please find the corresponding theoretical derivations related to the transformation and approximation in subsection \ref{OPSE1}, \ref{OPSE2}, and \ref{OPSE3}. The solution is presented in subsection \ref{OPSE4}.

\vspace{-.2cm}
\subsection{Transformation for Constraint (\ref{Prob3_con2})}
\label{OPSE1}
\vspace{-.1cm}

Constraint (\ref{Prob3_con2}) contains the non-convex multiplication term and thus is intractable for the optimization solvers. 
To make (\ref{Prob3_con2})  tractable, the logarithm expression on the left is removed at first by having the following constraint:
\vspace{-.2cm}
\begin{eqnarray}
\sum_m\sum_{n=1}^{|\mathcal{N}_{m}|}\tilde{I}_{m,n}^{i,j}G(\theta_{m,n}^{i,j})||{\bf{w}}_{m,n}||^2 u_{m,n}\leq {\bar{I}}_{th},~\forall c_{i,j}\in\mathcal{C}^{\text{PZ}},
\label{SIEOJTONCM}
\end{eqnarray}
where ${\bar{I}}_{th}=10^{(I_{th}/10-3)}$ is expressed in mWatts. Then, to decouple the term $G(\theta_{m,n}^{i,j})||{\bf{w}}_{m,n}||^2u_{m,n}$, we introduce an auxiliary variable $p_{m,n}$ by setting $||{\bf{w}}_{m,n}||^2\leq p_{m,n}$, which represents the power allocated to $u_{m,n}$. $p_{m,n}u_{m,n}$ can be transformed by having ${\bar{p}}_{m,n}=p_{m,n}u_{m,n}$, where $p_{m,n}^{\text{min}}\leq p_{m,n}\leq p_{m,n}^{\text{max}}$, and $p_{m,n}^{\text{min}}$ and $p_{m,n}^{\text{max}}$ are the lower and upper limit for the transmission power for $(m,n)$. We can further transform ${\bar{p}}_{m,n}$ according to the Big-M formulation \cite{wei2017optimal,ng2016power} by setting ${\tilde{p}}_{m,n}={\bar{p}}_{m,n}-p_{m,n}^{\text{min}}$. To this end, constraint (\ref{SIEOJTONCM}) can be transformed as follows
\begin{eqnarray}
&&{ {\bar{I}}_{th}- \sum_{m,n}\tilde{I}_{m,n}^{i,j}G(\theta_{m,n}^{i,j})\big({\tilde{p}}_{m,n}+{{u}}_{m,n}{{p}}_{m,n}^{\text{min}}\big)}{\geq {d}_{i,j}},\nonumber\\
&&~~\text{where }{|d_{i,j}|}{\leq D_{i,j}^{max}},\label{Trans_7b_1}\\
&&{{\tilde{p}}_{m,n}}{\leq {{u}}_{m,n}\big({{p}}_{m,n}^{\text{max}}-{{p}}_{m,n}^{\text{min}} \big),\forall ~m,n},\label{Trans_7b_2}\\
&&{{\tilde{p}}_{m,n}}{\geq p_{m,n}-(1-{{u}}_{m,n})\big({{p}}_{m,n}^{\text{max}}-{{p}}_{m,n}^{\text{min}} \big)},\label{Trans_7b_3}\\
&&{{\tilde{p}}_{m,n}\geq 0,}~{{\tilde{p}}_{m,n}}\leq p_{m,n}. \label{Trans_7b_4}
\end{eqnarray}
In (\ref{Trans_7b_1}), to enable a feasible expression for the relaxed transformation, we further relax this constraint by introducing a real variable $d_{i,j}$ with the range ${|d_{i,j}|}{\leq D_{i,j}^{max}}$. $d_{i,j}$ is equivalent to the cell feasibility indicator $v_{i,j}$ by simply having $v_{i,j}=\text{max}[\text{sgn}(d_{i,j}),0]$, where $\text{sgn}(\cdot)$ is a sign function extracting the sign of real variables. Constraints (\ref{Trans_7b_2})-(\ref{Trans_7b_4}), are imposed here to meet the requirement for the Big-M formulation.

Note that $G(\theta_{m,n}^{i,j})$ in (\ref{Trans_7b_1}) is still non-convex. By analyzing the detailed expression of $G(\theta_{m,n}^{i,j})$ in Appendix I, the following property holds for $G(\theta_{m,n}^{i,j})$ given any feasible point.
\vspace{-.5cm}
\begin{lem}
The inner approximation for antenna gain $G(\theta_{m,n}^{i,j})$ can be given as
\vspace{-.2cm}
\begin{eqnarray}
G(\theta_{m,n}^{i,j})\leq \frac{{\bf{w}}_{m,n}^H({\bf{v}}_{\theta})_{m,n}^{i,j}({{\bf{v}}_{\theta}}_{m,n}^{i,j})^H{\bf{w}}_{m,n}}{2\mathcal{R}\Big\{ (({\bf{w}}_{m,n}^l)^H\tilde{\bf{v}})^H({\bf{w}}_{m,n}^H\tilde{\bf{v}})\Big\} -|({\bf{w}}_{m,n}^l)^H\tilde{\bf{v}}|^2}. 
\label{FS11_RSYSY}
\end{eqnarray}
\label{FS11_GSHTY}
\end{lem}
\vspace{-.2cm}
\begin{proof}
The denominator in (\ref{FS11_RSYSY}) is derived based on the following results:
\begin{eqnarray}
{{\bf{w}}_{m,n}^H{\bf{A}}{\bf{w}}_{m,n}}&\geq&
2\mathcal{R}\Big\{ (({\bf{w}}_{m,n}^l)^H\tilde{\bf{v}})^H({\bf{w}}_{m,n}^H\tilde{\bf{v}})\Big\}\nonumber\\ 
&-&|({\bf{w}}_{m,n}^l)^H\tilde{\bf{v}}|^2,
\label{YUSUI_89SDE}
\vspace{-.1cm}
\end{eqnarray}
(\ref{FS11_RSYSY}) has a quadratic-over-linear expression, which is convex for the postive denominator term in (\ref{FS11_RSYSY}). Please find the expression of ${\bf{A}}$ in (\ref{SVCTR_34}) of Appendix I. Note that ${\bf{A}}=\tilde{\bf{v}}^H\tilde{\bf{v}}$ and $\tilde{\bf{v}}=({\bf{Q}}{\bf{\Gamma}}^{\frac{1}{2}})$, which can be derived from the fact that matrix ${\bf{A}}$ is a positive semi-definite (PSD) matrix and ${\bf{A}}=({\bf{Q}}{\bf{\Gamma}}^{\frac{1}{2}})({\bf{Q}}{\bf{\Gamma}}^{\frac{1}{2}})^H$. Here the matrix ${\bf{Q}}$ and ${\bf{\Gamma}}$ contain the eigenvectors and eigenvalues.
\end{proof}
\vspace{-.2cm}
We now use $G(\theta_{m,n}^{i,j})_{\text{inner}}$ to denote the right-hand term in (\ref{FS11_RSYSY}). Based on the transformed ${\tilde{p}}_{m,n}$ and Lemma \ref{FS11_GSHTY}, the multiplication term in (\ref{Prob3_con2}) is transformed into $G(\theta_{m,n}^{i,j})_{\text{inner}}~{\tilde{p}}_{m,n}$, which is still intractable and will be dealt with later.

\vspace{-.2cm}
\subsection{Transformation for Constraint (\ref{Prob3_con2_2})}
\label{OPSE2}

Constraint (\ref{Prob3_con2_2}) is the QoS requirement on each link $(m,n)$. This is a non-concave constraint as well and cannot be solved efficiently. To find a tractable formulation for (\ref{Prob3_con2_2}), we will develop a concave formulation for SU data rate at first and address the integer variable $u_{m,n}$ later. We propose an inner approximation bound for the SU data rate in Lemma \ref{STPOSEICE}.
\begin{lem}
For constraint (\ref{Prob3_con2_2}), the following inequality holds for any feasible point ${\bf{w}}_{m,n}^l$.
\vspace{-.25cm}
\begin{eqnarray}
\log_2\Big(1+\frac{|[\tilde{\bf{h}}_{m,n}^t]^H{\bf{w}}_{m,n}|^2}{I_{m,n}^{\text{inter}}+I_{m,n}^{\text{intra}}+\sigma_{m,n}^2}\Big) \geq f_2({\bf{W}}^{l},{\bf{W}}),
\label{MNVBSSHST_version1}\\
f_2({\bf{W}}^l,{\bf{W}})= -\frac{1}{\ln{2}}\frac{b({\bf{w}}_{m,n}^l)}{(a({\bf{W}}^l)+b({\bf{w}}_{m,n}^l))a({\bf{W}}^l)}\nonumber\\
\times (\frac{a({\bf{W}})}{v({\bf{w}}^l_{m,n},{\bf{w}}_{m,n})}-\frac{a({\bf{W}}^l)}{b({\bf{w}}_{m,n}^l)}).
\end{eqnarray}
\label{STPOSEICE}
\end{lem}

\begin{proof}
In (\ref{MNVBSSHST_version1}), $f_2({\bf{W}}^l,{\bf{W}})$ is the inner approximation for the data rate based on the first-order convexity condition where ${\bf{W}}^l=\{{\bf{w}}^l_{m,n}|~\forall~m\in\mathcal{M},~n\in\mathcal{N}_m\}$ and ${\bf{W}}=\{{\bf{w}}_{m,n}|~\forall~m\in\mathcal{M},~n\in\mathcal{N}_m\}$. The definitions of $a({\bf{W}})$, $b({\bf{w}}_{m,n}^l)$ and $v({\bf{w}}_{m,n}^l,{\bf{w}}_{m,n})$ are given by
\begin{eqnarray}
&&a({\bf{W}})=I_{m,n}^{\text{inter}}+I_{m,n}^{\text{intra}}+\sigma_{m,n}^2,\nonumber\\
&&~~b({\bf{w}}_{m,n}^l)=|[\tilde{\bf{h}}_{m,n}^t]^H{\bf{w}}_{m,n}^l|^2,\\
&&v({\bf{w}}_{m,n}^l,{\bf{w}}_{m,n})=-|[\tilde{\bf{h}}_{m,n}^t]^H{\bf{w}}_{m,n}^l|^2\nonumber\\ &&+2\mathcal{R}\Big\{ ({\bf{w}}_{m,n}^l)^H[\tilde{\bf{h}}_{m,n}^t][\tilde{\bf{h}}_{m,n}^t]^H{\bf{w}}_{m,n} \Big\},
\label{RTST_2}
\end{eqnarray}
It is easy to observe that $a({\bf{W}})$ has a quadratic form and $v({\bf{w}}_{m,n}^l,{\bf{w}}_{m,n})$ is a linear term. Therefore, we can prove the approximation in (\ref{MNVBSSHST_version1}) holds. 
\end{proof}

\vspace{-.2cm}
\subsection{Transformation for the Integer Variable $u_{m,n}$}
\label{OPSE3}

(\ref{Prob3_con5}) has a binary constraint on variable ${{u}}_{m,n}\in\{0,1\}$, which leads to a complex combinatorial optimization problem.
To reduce complexity, the similar techniques in \cite{ng2016power} are adopted to relax the binary variable into a real variable $\bar{{u}}_{m,n}\in[0,1]$, which is subject to the following constraints:
\vspace{-0.3cm}
\begin{eqnarray}
&& 0\leq{\bar{u}}_{m,n}\leq 1,\label{C4a_version1}\\
&& \sum_m\sum_n{\bar{u}}_{m,n}-\sum_m\sum_n{\bar{u}}_{m,n}^2\leq 0, ~\forall ~u_{m,n}\in{\bf{U}},\label{C4b_version1}
\end{eqnarray}
where ${\bar{u}}_{m,n}$ is a real variable after relaxation and (\ref{C4b_version1}) forces $u_{m,n}$ to approach to 0 or 1. Please note the left-hand term of the inequality in (\ref{C4b_version1}) is non-convex. In the following, we adopt the successive convex approximation techniques to find its global lower bound. For instance, let $f({\bar{u}}_{m,n})=\sum_m\sum_n{\bar{u}}_{m,n}-\sum_m\sum_n{\bar{u}}_{m,n}^2$. With given location point ${\bar{u}}_{m,n}^l$ in the $l$-th iteration, we can obtain the following lower bound for $f({\bar{u}}_{m,n})$:
\begin{eqnarray}
f_3({\bar{u}}_{m,n}^l,{\bar{u}}_{m,n})= f({\bar{u}}_{m,n})\nonumber\\
+\{\sum_m\sum_n(1-2{\bar{u}}_{m,n})\}({\bar{u}}_{m,n}-{\bar{u}}_{m,n}^{~l}).
\label{RTSYSY_version1}
\end{eqnarray}

$f_3({\bar{u}}_{m,n}^l,{\bar{u}}_{m,n})$ can be incorporated into the objective function (\ref{Prob3_obj}) by augmenting $\eta f_3({\bar{u}}_{m,n}^l,{\bar{u}}_{m,n})$ into the objective function, where $\eta << -1$ is a negative penalty factor.

\vspace{-.2cm}
\subsection{Transformation for Objective (\ref{Prob3_obj})}
\label{OPSE4}

Objective (\ref{Prob3_obj}) contains the non-convex SU data rate $\log_2(\cdot)$ and the multiplication term $u_{m,n} \log_2(\cdot)$, which have to be tackled. For the SU data rate, we can directly utilize the approximation for the SU data rate in \textbf{Lemma \ref{STPOSEICE}}. The multiplication term $u_{m,n} \log_2(\cdot)$ will be dealt with later. Please note besides the original SU data rate, the transformed objective has to contain several terms developed for the aforementioned transformation, i.e., the penalty term $\xi \sum_{i,j} \big[\max(0,-d_{i,j})\big]^2$ and the penalty for the binary relaxation $\eta \sum_{m,n} f_3({\bar{u}}_{m,n}^l,{\bar{u}}_{m,n})$.

Until now, most of non-convex/concave constraints except the coupled multiplication terms in (\ref{SIEOJTONCM}) have been transformed. However,
we notice that when either ${\bf{W}}$ or the remaining variables $\{\tilde{\bf{P}},{\bf{P}},\bar{\bf{U}},{\bf{D}}\}$ is fixed, where $\tilde{\bf{P}}=\{\tilde{{p}}_{m,n} ~|~m\in\mathcal{M},n\in\mathcal{N}_m\}$, ${\bf{P}}=\{{{p}}_{m,n} ~|~m\in\mathcal{M},n\in\mathcal{N}_m\}$, $\bar{\bf{U}}=\{\bar{{u}}_{m,n} ~|~m\in\mathcal{M},n\in\mathcal{N}_m\}$, and ${\bf{D}}=\{{{d}}_{i,j} ~|~(i,j)\in\mathcal{C}_c\}$, the problem can be efficiently solved.
Motivated by this observation, we propose an alternating optimization (AO) based algorithm to
iteratively optimize ${\bf{W}}$ or the variables $\{\tilde{\bf{P}},{\bf{P}},\bar{\bf{U}},{\bf{D}}\}$ with the other being fixed until the convergence.

\paragraph{Optimizing $\{\tilde{\bf{P}},{\bf{P}},\bar{\bf{U}},{\bf{D}}\}$ for given ${\bf{W}}$} For the ease of presentation, for any constraint or function containing the fixed variables ${\bf{W}}$, we denote it as $\widehat{(\cdot)}$.
\begin{subequations}
\label{Prob32_all}
\begin{align}
{\text{(P1.1)}} \mathop {\max }\limits_{\tilde{\bf{P}},{\bf{P}},\bar{\bf{U}},{\bf{D}}}\quad & \sum_{m,n}\bar{u}_{m,n} \widehat{f_2}({\bf{W}}^l,{\bf{W}}) +\eta \sum_{m,n} f_3({\bar{u}}_{m,n}^l,{\bar{u}}_{m,n})\nonumber\\
&+{{\xi\sum_{i,j} \big[\max(0,-d_{i,j})\big]^2}}\nonumber\\
\text{s.t.}\quad
& {\bar{I}}_{th}- \sum_{m,n}\tilde{I}_{m,n}^{i,j}\widehat{G}(\theta_{m,n}^{i,j})_{\text{inner}}\nonumber\\ &\times~\big({\tilde{p}}_{m,n}+\bar{u}_{m,n}{{p}}_{m,n}^{\text{min}}\big){\geq {d}_{i,j}}, \label{WEIOTSCQ}\\
& \widehat{f}_2({\bf{W}}^{l},{\bf{W}})\geq \bar{u}_{m,n}r_{m,n},\label{SMSIOESEIO}\\
& \sum_{n\in\mathcal{N}_m}\tilde{p}_{m,n}\leq P_{m},~\forall m\in\mathcal{M},\label{POSIEOT}\\
&(\ref{Trans_7b_2})-(\ref{Trans_7b_4}),~(\ref{C4a_version1}).
\end{align}
\end{subequations}
\paragraph{Optimizing ${\bf{W}}$ for given $\{\tilde{\bf{P}},{\bf{P}},\bar{\bf{U}},{\bf{D}}\}$} For any constraint or function containing the fixed variables $\tilde{p}_{m,n}$, $p_{m,n}$, $\bar{u}_{m,n}$, and $d_{i,j}$, we denote it as $\widehat{(\cdot)}$.
\begin{subequations}
\label{Prob32_all}
\begin{align}
{\text{(P1.2)}~} \mathop {\max }\limits_{\bf{W}}\quad & \sum_{m,n}\widehat{\bar{u}}_{m,n} {f_2}({\bf{W}}^l,{\bf{W}}) +\eta \sum_{m,n} \widehat{f_3}({\bar{u}}_{m,n}^l,{\bar{u}}_{m,n})\nonumber\\ 
&+{{\xi\sum_{i,j} \big[\max(0,-\widehat{d}_{i,j})\big]^2}}\nonumber\\
\text{s.t.}\quad
& {\bar{I}}_{th}- \sum_{m,n}\tilde{I}_{m,n}^{i,j}{G}(\theta_{m,n}^{i,j})_{\text{inner}}\nonumber\\
&\times~\big(~\widehat{\tilde{p}}_{m,n}+\widehat{\bar{u}}_{m,n}{{p}}_{m,n}^{\text{min}}\big){\geq \widehat{d}_{i,j}},\label{WEIOTSCQ_p12}\\
& {f}_2({\bf{W}}^{l},{\bf{W}})\geq \widehat{\bar{u}}_{m,n}r_{m,n},\label{SMSIOESEIO_p12}\\
& \sum_{n\in\mathcal{N}_m}\widehat{\tilde{p}}_{m,n}\leq P_{m},~\forall m\in\mathcal{M},\label{POSIEOT_p12}\\
&\widehat{(\ref{Trans_7b_2})}-\widehat{(\ref{Trans_7b_4})},~\widehat{(\ref{C4a_version1})}.
\end{align}
\end{subequations}
\vspace{-.15cm}
The overall iterative algorithm to solve (\ref{Prob3_all}) is given in Algorithm 1. 
\setlength{\textfloatsep}{0pt}
\begin{algorithm}[tb]
\caption{SU Data Rate Maximization Given the PU Privacy Constraint}
\label{alg1_AO}
\begin{algorithmic}
\STATE \textbf{Input}: $r_{m,n}$, $\tilde{\bf{h}}_{m,n}^t$, $\mathcal{M}$, $\mathcal{N}_m$,$\sigma_{m,n}^2$, $P_m$, $\eta$, $\xi$, convergence tolerance $\epsilon$, maximum iteration number $L_{max}$, privacy zone $\mathcal{C}_r^{\text{PZ}}$.
\STATE \textbf{Output}: $\tilde{\bf{P}},{\bf{P}},\bar{\bf{U}}$, ${\bf{W}}$.
\STATE \textbf{step-1}: Initialize ${\bf{W}}^{(0)}$, $[{\bar{u}}_{m,n}^l]^{(0)}$, $[{\bf{W}}^l]^{(0)}$, counter $l=1$.
\STATE \textbf{step-2}: $\textbf{Repeat:}$
\STATE \qquad \textbf{step-3} $\textbf{Repeat:}$\ \
\STATE \qquad \qquad For given ${f_2}([{\bf{W}}^l]^{(l-1)},{\bf{W}}^{(l-1)})$, and ${G}(\theta_{m,n}^{i,j})^{(l-1)}$, solve problem (P1.1).\\
\STATE \qquad \qquad Update $\bar{u}_{m,n}^{(l)}=\bar{u}_{m,n}$, $f_3([{\bar{u}}_{m,n}^l]^{(l)},{\bar{u}}_{m,n}^{(l)})$, $[{\bar{u}}_{m,n}^l]^{(l)}={\bar{u}}_{m,n}$. \\
\STATE \qquad $\textbf{Until:}$ the objective value in (P1.1) reaches $\epsilon$.\\
\STATE \qquad \textbf{step-4} Update $\tilde{p}_{m,n}^{(l)}=\tilde{p}_{m,n}$, $p_{m,n}^{(l)}=p_{m,n}$, $\bar{u}_{m,n}^{(l)}=\bar{u}_{m,n}$, ${d}_{i,j}^{(l)}={d}_{i,j}$.\\
\STATE \qquad \textbf{step-5} $\textbf{Repeat:}$\ \
\STATE \qquad \qquad For given $\tilde{p}_{m,n}^{(l)}$, $p_{m,n}^{(l)}$, $\bar{u}_{m,n}^{(l)}$, ${d}_{i,j}^{(l)}$, and $f_3([{\bar{u}}_{m,n}^l]^{(l)},{\bar{u}}_{m,n}^{(l)})$, solve problem (P1.2).\\
\STATE \qquad \qquad Update ${\bf{W}}^{(l)}={\bf{W}}$, ${f_2}([{\bf{W}}^l]^{(l)},{\bf{W}}^{(l)})$, ${G}(\theta_{m,n}^{i,j})^{(l)}$, $[{\bf{W}}^l]^{(l)}={\bf{W}}$. \\
\STATE \qquad $\textbf{Until:}$ the objective value in (P1.2) reaches convergence tolerance $\epsilon$.\\
\STATE \qquad \textbf{step-6} Update $l=l+1$.\\
\STATE  $\textbf{Until:}$ the objective value reaches convergence.
\end{algorithmic}
\end{algorithm}

\vspace{-.2cm}
\subsection{Convergence Analysis}

To analyze the overall convergence of Algorithm 1, we need to analyze the convergence of (P1.1) and (P1.2) individually, and then give a comprehensive analysis for the overall AO based algorithm. Let's use $F^{\text{~P1.1}}(\cdot)$ and $F^{\text{~P1.2}}(\cdot)$ to denote the objective functions of (P1.1) and (P1.2). $F_O^{\text{~P1.1}}(\cdot)$ and $F_O^{\text{~P1.2}}(\cdot)$ represent the original objective functions of (P1.1) and (P1.2). We have to point out $F_O^{\text{~P1.1}}(\cdot)=F_O^{\text{~P1.2}}(\cdot)$ for the same input.

{\begin{lem}
For the iteration of SCA algorithm in step 3, each iteration for $F_O^{\text{ P1.1}}({\bf{\Theta}},\bar{u}_{m,n}^l)$ of (P1.1) has the following property:
\vspace{-.4cm}
\begin{eqnarray}
F_O^{\text{ P1.1}}({\bf{\Theta}},\bar{u}_{m,n}^l)=F^{\text{ P1.1}}({\bf{\Theta}},\bar{u}_{m,n}^l,\bar{u}_{m,n}^l)\nonumber\\
\leq F^{\text{ P1.1}}({\bf{\Theta}},\bar{u}_{m,n}^l,\bar{u}_{m,n}^{l+1})\leq F_O^{\text{ P1.1}}({\bf{\Theta}},\bar{u}_{m,n}^{l+1}).
\label{POSEIOTSJM}
\end{eqnarray}
\end{lem}
\begin{proof}
Define ${\bf{\Theta}}=\{\tilde{p}_{m,n},p_{m,n},d_{i,j}\}$, where $m\in\mathcal{M}, ~n\in\mathcal{N}_m$, and $(i,j)\in\mathcal{C}^{\text{PZ}}$. For the feasible point $\bar{u}_{m,n}^{l+1}$ and $\bar{u}_{m,n}^{l}$, the searching result $F^{\text{ P1.1}}({\bf{\Theta}},\bar{u}_{m,n}^l,\bar{u}_{m,n}^{l+1})$ is better than $F^{\text{ P1.1}}({\bf{\Theta}},\bar{u}_{m,n}^l,\bar{u}_{m,n}^l)$ for (P1.1), whenever $\bar{u}_{m,n}^{l+1}\neq\bar{u}_{m,n}^{l}$.  On the other hand, if $\bar{u}_{m,n}^{l+1}=\bar{u}_{m,n}^{l}$, i.e. $\bar{u}_{m,n}^{l}$ is the optimal solution of (P1.1), then it must satisfy the first order necessary optimality condition. The value of $F_O^{\text{ P1.1}}({\bf{\Theta}},\bar{u}_{m,n}^l)$ can be improved iteratively. We thus can prove that the sequence $\bar{u}_{m,n}^{l+1}$ converges to a point
satisfying the first order necessary optimality condition. In a similar way, we can prove the convergence of (P1.2) based on SCA algorithm.
\end{proof}}

{
\begin{lem}
Let's partition the variables in (P1.1) and (P1.2) into two blocks, i.e., $\Gamma=\{\bf{W}\}$, $\Theta=\{\tilde{p}_{m,n},p_{m,n},
\bar{u}_{m,n},d_{i,j}\}$, where $n\in\mathcal{N}_m,~m\in\mathcal{M}$ and $(i,j)\in\mathcal{C}^{\text{PZ}}$. We can prove the SU network throughput will increase after each iteration:
\vspace{-.3cm}
\begin{eqnarray}
F_o^{\text{ P1.1}}(\Theta^{r+1},\Gamma^{~r})\overset{\mathrm{a}}{=}F_o^{\text{ P1.2}}(\Theta^{r+1},\Gamma^{~r})\nonumber\\
\overset{\mathrm{b}}
{\leq}F_o^{\text{ P1.2}}(\Theta^{r+1},\Gamma^{~r+1}).
\label{SNMEIS_2}
\end{eqnarray}
\end{lem}
\begin{proof}
After step 3, since the optimal solution of (P1.1) is obtained for given $\Theta^r$ and $\Gamma^r$, we have $F_o^{\text{ P1.1}}(\Theta^{r},\Gamma^r)\leq F_o^{\text{ P1.1}}(\Theta^{r+1},\Gamma^r)$.
(a) in (\ref{SNMEIS_2}) holds since  $F_O^{\text{~P1.1}}(\cdot)=F_O^{\text{~P1.2}}(\cdot)$ for the same input. (b) in (\ref{SNMEIS_2}) holds since the first-order Taylor expansions in (P1.1) is tight at the given feasible points. Finally, we can observe that problem (P1.2) at $\{\Theta^{r+1},\Gamma^{~r+1}\}$ has the same objective value as
that of problem (P1.1). As a result, SU network throughput will increase by iteratively running the alternating optimization algorithm.
\end{proof}}

\vspace{-.2cm}
\subsection{Computational Complexity Analysis}

By adopting the analysis in \cite{zhu2017beamforming,zhu2016outage}, we present the complexity for Algorithm 1. Let $N^{\text{node}}$, $|\mathcal{M}|$, and $|\mathcal{C}^{\text{PZ}}|$ denote the number of SR nodes, the number of ST nodes, the size of privacy zone. For (P1.1), the major complexity arises due to the linear matrix inequality (LMI) contraints. We use $Q_{\text{max}}^1$ to denote the SCA iteration number. To solve (P1.1), the computational complexity can be given as $\mathcal{O}(n_1Q_{\text{max}}^1(\sqrt{7N^{\text{node}}+ |\mathcal{M}|+|\mathcal{C}^{\text{PZ}}|}((7N^{\text{node}}+ |\mathcal{M}|+|\mathcal{C}^{\text{PZ}}|)(n_1+1)+n_1^{~2})))$, where $n_1=\mathcal{O}(3N^{\text{node}}+|\mathcal{C}^{\text{PZ}}|)$. In (P1.2), the computational complexity is due to the second order cone (SOC) constraint. We use $Q_{\text{max}}^2$ denotes the SCA iteration number for (P1.2). The computational complexity for (P1.2) can be expressed as $\mathcal{O}(n_2Q_{\text{max}}^2(\sqrt{2(N^{\text{node}}+|\mathcal{C}^{\text{PZ}}|)}(N^{\text{node}}+|\mathcal{C}^{\text{PZ}}|+n_2^{~2})))$, where $n_2=\mathcal{O}(N_tN^{\text{node}})$.

%% file: maximizeOP.tex
\vspace{-.2cm}
\section{SU Network Beamforming for PU Location Privacy Enhancement}
\label{Prob_formulation}
\vspace{-.1cm}

In the previous section, a problem is formulated to maximize the SU network throughput subject to the PU privacy contraints.   
In this section, we aim to enhance the PU location privacy given the SU network throughput requirement.
A two-step scheme will be proposed first, followed by the convergence discussion and complexity analysis.

The benefits of beamforming 
for enhancing the PU location privacy would appear to be straightforward at first glance. In fact, beamforming changes rapidly along with the fast fading wireless channels, whereas the privacy zone holds static for a long time period. Therefore, the existing beamforming techniques cannot be directly applied in our problem.
Here we begin by proposing a scheme to jointly optimize the privacy zone and beamforming design in subsection \ref{CSIEOT}. 
We then identify the key challenges in this joint design problem. 
To efficiently solve this problem, a 2-step scheme is proposed in subsections \ref{TYSYU_tuy8} and \ref{STEP_2}.

\vspace{-.2cm}
\subsection{Joint Design for Privacy Zone and SU Network Beamforming}
\label{CSIEOT}

For given CSI and privacy zone candidate set, the PU location privacy $\log_2(|\mathcal{C}^{\text{PZ}}|)$ can be maximized by jointly optimizing the beamformer ${\bf{W}}$, the power limit ${\bf{P}}$, SU node exclusion ${\bf{U}}$, and binary indicator ${\bf{V}}$, which yields the following problem

\vspace{-.5cm}
\begin{subequations}
\label{P1_SDFSG_prob}
\begin{align}
\mathop {\max }\limits_{{\bf{V}},{\bf{W}},{\bf{P}},{\bf{U}}}\quad & \log_2(\sum_{i,j}v_{i,j})
\label{UIS1sSETY_KIO_prob}\\
\text{s.t.}\quad 
&v_{i,j}\sum_{m,n}\tilde{I}_{m,n}^{i,j}G(\theta_{m,n}^{i,j})p_{m,n} u_{m,n}~\leq~ \tilde{I}_{th},\label{communication_utility_11}\\
&v_{i,j}\in\{0,1\} ,~{c_{i,j}\in\mathcal{C }_c},~u_{m,n}\in\{0,1\},
\label{LOSETDX_SET_prob}\\
&{\sum_m\sum_n u_{m,n} \geq \sum_m|\mathcal{N}_{m}|-N_{re}},\label{OPISO_SAET}\\
&(\ref{Prob3_con2_2})-(\ref{Prob3_con5}),~m\in\mathcal{M},~n\in\mathcal{N}_m.
\end{align}
\end{subequations}
\vspace{-.1cm}
In problem (\ref{P1_SDFSG_prob}), the PU location privacy is optimized by designing the privacy zone $\mathcal{C}^{\text{PZ}}$, which contains the cells selected from the candidate set $\mathcal{C }_c$. 
For each cell $c_{i,j}\in\mathcal{C }_c$, a binary variable $v_{i,j}\in\{0,1\}$ is introduced to denote whether this cell is included in the privacy zone, such that  $v_{i,j}=1$ when cell $c_{i,j}$ is included; otherwise, not.
Therefore, the size of privacy zone is given as $|\mathcal{C}^{\text{PZ}}|=\sum_{i,j}v_{i,j}$ in (\ref{UIS1sSETY_KIO_prob}).

Problem (\ref{P1_SDFSG_prob}) shares several constrains as those in Problem (\ref{Prob3_all}), i.e., the constraints (\ref{Prob3_con2_2})-(\ref{Prob3_con5}).
In (\ref{communication_utility_11}), a binary indicator $v_{i,j}$ is 
multiplied by the terms on the left side to make the inequality always hold. Note constraint (\ref{OPISO_SAET}) is introduced to set a limit on the number of excluded SU nodes, which avoids excluding all of SU nodes on the band while expanding the privacy zone. 

Unfortunately, there arises an issue while solving problem (\ref{P1_SDFSG_prob}): including the rapidly-changing CSI $\tilde{\bf{h}}_{m,n}^t$ in the formulation, which is necessary for the beamforming, leads to a fast-changing privacy zone. While utilizing the beamforming for the privacy zone design, each beamformer requires a very short updating period in response to the fast fading CSI $\tilde{\bf{h}}_{m,n}^t$, in the scale of microseconds to milliseconds. 
In consequence, involving the rapidly-changing CSI $\tilde{\bf{h}}_{m,n}^t$ in (\ref{P1_SDFSG_prob}) to search a feasible privacy zone will force the privacy zone to be frequently changed to keep pace with the CSI $\tilde{\bf{h}}_{m,n}^t$. However, the privacy zones developed for the PUs usually have a long updating cycle, i.e., privacy zone has to be activated and then held still for a period of time while PUs are in operation, i.e., tens of minutes or even a few hours. 
Based on the above observations, the CSI $\tilde{\bf{h}}_{m,n}^t$ cannot be directly included in the problem formulation.
That is, the impact of fast-varying beamforming on the privacy zone design must be evaluated in a long-term manner.

To address the aforementioned issue, we decouple problem (\ref{P1_SDFSG_prob}) into two steps: first, the long-term privacy zone design is performed in step 1. The short-term beamforming information including the CSI $\tilde{\bf{h}}_{m,n}^t$ is discarded in this step and we instead utilize an upper bound for the beamformer gain to evaluate its long-term impact. Second, the short-term beamforming design is conducted later in step 2 based on the privacy zone design in step 1.

%% file: Cloak_Region_Design.tex
\vspace{-.3cm}
\subsection{Step 1: Improving Location Privacy Based on Network Deployment Information}
\label{TYSYU_tuy8}

After dropping the rapidly-varying CSI $\tilde{\bf{h}}_{m,n}^t$, the PU location privacy can be optimized by solving the following problem.
\vspace{-.2cm}
\begin{subequations}
\label{P1_SDFSG}
\begin{align}
\mathop {\max }\limits_{{\bf{V}},{\bf{P}},{\bf{U}}}\quad & \log_2(\sum_i\sum_j v_{i,j})
\label{UIS1sSETY_KIO}\\
\text{s.t.}\quad 
&v_{i,j}\sum_{m,n}\tilde{I}_{m,n}^{i,j}G(\theta_{m,n}^{i,j})p_{m,n} u_{m,n}~\leq~ \tilde{I}_{th},
\label{Prob1_cons1}\\
&{\sum_{m,n} u_{m,n}{p}_{m,n}}{\leq P_{m,n}},{u_{m,n}\in\{0,1\},~v_{i,j}\in\{0,1\} },
\label{MKsCGST_SAET}\\
&{\sum_m\sum_n u_{m,n} \geq \sum_m|\mathcal{N}_{m}|-N_{re}},
\label{OPISO_SAET}
\end{align}
\end{subequations}

Even though the constraints related to the registered network deployment information are only needed in (\ref{P1_SDFSG}), such as the maximum number of SRs $N_{m}^{~\text{max}}$, and the antenna attributes of each multi-antenna STs, (\ref{Prob1_cons1}) still contains the antenna gain, which is parameterized by the fast-varying beamforming information. 
To evaluate the impact of beamforming on the aggregate interference in (\ref{Prob1_cons1}) from a long-term perspective, we propose an upper bound for the antenna gain, which has a deterministic expression.

\subsubsection{Transformation for the Constraint (\ref{Prob1_cons1})}

The deterministic upper bound for the antenna gain is given as follows.
\vspace{-.5cm}
\begin{lem}
For the uniform linear array (ULA) with isotropic elements, the normalized antenna gain has the upper bound as:
\vspace{-.45cm}
\begin{eqnarray}
G(\theta_{m,n}^{i,j})\leq G({\bf{\lambda}})\overset{\Delta}{=} \epsilon_a\frac{(\tilde{\lambda}_{m,n}^1)}{(\lambda_{2})_{A}},
\label{ERTS_erwet}
\end{eqnarray}
where $\tilde{\lambda}_{m,n}^1$ is the maximum eigen-value for ${\bf{B}}={\bf{v}}(\theta_k){\bf{v}}(\theta_k)^H$ and $(\lambda_{~\text{min}})_{A}$ is the minimum eigen-value for the matrix ${\bf{A}}$.
\label{SSERT_34}
\end{lem}
\vspace{-.2cm}
\begin{proof}
For the ease of presentation, we have moved the proof for Lemma \ref{SSERT_34} to Appendix II.
\end{proof}
\vspace{-.3cm}
Based on the upper bound in Lemma \ref{SSERT_34}, we thus can formulate the upper bound for the aggregate interference, which can be expressed as:
\vspace{-.2cm}
\begin{eqnarray}
~v_{i,j}\sum_{m,n}\tilde{I}_{m,n}^{i,j}G({\bf{\lambda}})p_{m,n}u_{m,n}~\leq~ \tilde{I}_{th}.
\label{NMJIOP_version1}
\end{eqnarray}
Relying on (\ref{NMJIOP_version1}), we can bound the influence of beamforming so that it can be evaluated in a long-term manner. 
The upper bound is raltively too loose for the SU with beamforming capability. A revised upper bound $G(\theta_{m,n}^{\text{ low}})=r_{\text{ratio}}G({\bf{\lambda}})~,r_{\text{ratio}}\in(0,1)$ is considered here. Note that the left-hand side of (\ref{NMJIOP_version1}) is non-convex due to the multi-linear production term $p_{m,n}u_{m,n}$ and the binary variable $v_{i,j}$. In general, there is no efficient method to obtain the optimal solution. In what follows, to address the non-convexity, the multiplication term $p_{m,n}u_{m,n}$ is transformed via introducing the auxiliary $\tilde{p}_{m,n}=p_{m,n}u_{m,n}$ by using the Big-M formulation \cite{wei2017optimal,ng2016power} techniques developed in (\ref{Trans_7b_2})-(\ref{Trans_7b_4}). For the binary variable $v_{i,j}$ and $u_{m,n}$, the similar relaxation technique introduced in (\ref{C4a_version1})-(\ref{C4b_version1}) can be adopted. Please note the binary variable $v_{i,j}$ requires a constraint $f_3({\bar{v}}_{i,j}^l,{\bar{v}}_{i,j})\leq 0$, where $f_3({\bar{v}}_{i,j}^l,{\bar{v}}_{i,j})=f({\bar{v}}_{i,j})+\{\sum_{i,j}(1-2{\bar{v}}_{i,j})\}({\bar{v}}_{i,j}-{\bar{v}}_{i,j}^{~l})$ and $f({\bar{v}}_{i,j})=\sum_{i,j}{\bar{v}}_{i,j}-\sum_{i,j}{\bar{v}}_{i,j}^2$. Finally, the modified problem can be expressed as

\vspace{-.5cm}
\begin{subequations}
\label{P253_SDFSG}
\begin{align}
\mathop {\max }\limits_{\tilde{\bf{P}},{\bf{P}},\bar{\bf{U}},\bar{\bf{V}}}\quad & \log_2(\sum_{i,j} \bar{v}_{i,j})+\eta\sum_{i,j}f_3({\bar{v}}_{i,j}^l,{\bar{v}}_{i,j})\nonumber\\
&+\eta\sum_{m,n}f_3({\bar{u}}_{m,n}^l,{\bar{u}}_{m,n})
\label{sSSE_45}\\
\text{s.t.}\quad 
&\bar{v}_{i,j} \sum_{m,n}\tilde{I}_{m,n}^{i,j}G(\theta_{m,n}^{\text{ low}})\big({\tilde{p}}_{m,n}+\bar{{u}}_{m,n}{{p}}_{m,n}^{\text{min}}\big)~{\leq~\bar{I}_{th}},
\label{C1}\\
& 0\leq \bar{v}_{i,j} \leq 1,~ 0\leq \bar{u}_{m,n} \leq 1,\label{C1_1}\\
& (\ref{Trans_7b_2})-(\ref{Trans_7b_4}),~ \sum_{m,n}\tilde{p}\leq P_m,\label{C1_2}\\
& {\sum_m\sum_n \bar{u}_{m,n} \geq \sum_m|\mathcal{N}_{m}|-N_{re}}.
\label{C2}
\end{align}
\end{subequations}

We augment $f_3({\bar{u}}_{m,n},{\bar{u}}_{m,n}^l)$ into the objective function via a negative penalty factor $\eta << -1$. Unfortunately, the problem above is still intractable due to the non-convex multiplication term in (\ref{C1}). AO algorithm can be applied here to iteratively optimizing $\bar{v}_{m,n}$ or rest variables $\{\tilde{p}_{m,n},{p}_{m,n}, \bar{u}_{m,n}\}$ with the other fixed. The output of (\ref{P253_SDFSG}) are the power allocation $p_{m,n}\in \mathcal{P}^{\text{step 1}}$, the SU node index $\bar{u}_{m,n}$ with $\bar{u}_{m,n}=1$, collected in $\mathcal{U}^{\text{step 1}}$, and the privacy zone design $v_{i,j}\in\mathcal{V}^{\text{step 1}}$. 
\vspace{-.2cm}
\subsection{Computational Complexity Analysis}

Let $N^{\text{node}}$ and $|\mathcal{C}_C|$ represent the number of SR nodes and the size of the candidate set. First, for given $v_{i,j}^r$ in the $r$-th iteration, solving (\ref{P253_SDFSG}) requires a computational complexity as $\mathcal{O}(nQ_{\text{max}}\sqrt{8N^{\text{node}}}(8N^{\text{node}}(n+1)+n^2)$, where $n=\mathcal{O}({3N^{\text{node}}})$ and $Q_{\text{max}}$ is the SCA iteration number. Second, given $\{\tilde{p}_{m,n}^r,{p}_{m,n}^r, \bar{u}_{m,n}^r\}$ in the $r$-th iteration, to solve (\ref{P253_SDFSG}), the computational complexity is given as $\mathcal{O}(nQ_{\text{max}}\sqrt{3|\mathcal{C}_C|}(3|\mathcal{C}_C|(n+1)+n^2)$, where $n=\mathcal{O}({3N^{\text{node}}})$ and $Q_{\text{max}}$ is the SCA iteration number.

%% file: Beamformer_design.tex
\vspace{-.2cm}
\subsection{Step 2: Further Improving SU Data Rate Via Limited Antenna Gain}
\label{STEP_2}

For any given privacy zone, power allocation, and SU node selection, i.e., ${\mathcal{P}^{\text{step 1}},\mathcal{U}^{\text{step 1}},\mathcal{V}^{\text{step 1}}}$, the SU data rate can be improved by further optimizing the beamforming vector ${\bf{w}}_{m,n}$ and the problem can be formulated as
\vspace{-.1cm}
\begin{subequations}
\label{Step_2}
\begin{align}
\mathop {\max }\limits_{{\bf{U}},{\bf{W}}}\quad & \sum_m\sum_n u_{m,n}\log_2\Big(1+\frac{|[\tilde{\bf{h}}_{m,n}^t]^H{\bf{w}}_{m,n}|^2}{I_{m,n}^{\text{inter}}+I_{m,n}^{\text{intra}}+\sigma_{m,n}^2}\Big)
\label{STEP2_obj}\\
\text{s.t.}\quad 
&\log_2\Big(1+\frac{|[\tilde{\bf{h}}_{m,n}^t]^H{\bf{w}}_{m,n}|^2}{I_{m,n}^{\text{inter}}+I_{m,n}^{\text{intra}}+\sigma_{m,n}^2}\Big)\geq u_{m,n}r_{m,n},\label{STEP2_con1}\\
&\sum_{m,n}\tilde{I}_{m,n}^{i,j}G(\theta_{m,n}^{i,j})p_{m,n}u_{m,n}\nonumber\\
&\leq \sum_{m,n}\tilde{I}_{m,n}^{i,j}G(\theta_{m,n}^{\text{ low}})p_{m,n}u_{m,n} ,\label{STEP2_con2}\\
&{||{\bf{w}}_{m,n}||^2}{\leq p_{m,n}},
\label{STEP2_con3}\\
&u_{m,n}\in\{0,1\},~v_{i,j}\in \mathcal{V}^{\text{step 1}}, u_{m,n}\in \mathcal{U}^{\text{step 1}}.
\label{STEP2_con4}
\end{align}
\end{subequations}
In (\ref{Step_2}), we consider the constraints on the minimum QoS requirement in (\ref{STEP2_con1}), the constraints on the aggregated interference over each cell in (\ref{STEP2_con2}) and the constraints on the transmission power for the beamformer in (\ref{STEP2_con3}). In (\ref{STEP2_con2}), the short-term antenna gain has been bounded according to Lemma 2 in step 1. 

{In existing works \cite{nasir2016secrecy,nasir2017secure}, beamforming is usually addressed by iterative algorithms such as successive convex approximation (SCA) or alternating direction method of multipliers (ADMM). In this paper, we have to point out that (\ref{Step_2}) has non-convex constraint (\ref{STEP2_con2}). In general, there is no standard method for solving such nonconvex optimization problems efficiently \cite{nasir2016secrecy,nasir2017secure}. In the following, the similar transformation tricks developed to solve (5) are adopted.}

\subsubsection{Transformations for Constraint (\ref{STEP2_con2})}
The constraint (\ref{NMJIOP_version1}) on the aggregate interference contains the upper bound for the antenna gain. Because of the fact that (\ref{NMJIOP_version1}) is formulated involving the SU nodes in $\mathcal{U}^{\text{step 1}}$ and $p_{m,n}\in \mathcal{P}^{\text{step 1}}$, this constraint can be equivalently converted to the following form:
\vspace{-.3cm}
\begin{eqnarray}
\sum_{m,n}\tilde{g}_{m,n}^{i,j}G(\theta_{m,n}^{i,j})u_{m,n}&\leq& \sum_{m,n}\tilde{g}_{m,n}^{i,j}G(\theta_{m,n}^{\text{ low}})u_{m,n},
\label{OPWERTZ}\\
||{\bf{w}}_{m,n}||_2^2&\leq& p_{m,n},
\label{YUSIE}
\end{eqnarray}
where $\tilde{g}_{m,n}^{i,j}=\tilde{I}_{m,n}^{i,j}p_{m,n}$. Note $G(\theta_{m,n}^{i,j})$ and the production term $G(\theta_{m,n}^{i,j})u_{m,n}$ are non-convex. In subsequent lemma, we adopt the successive convex optimization technique to get an inner approximation for the antenna gain first. The term $G(\theta_{m,n}^{i,j})u_{m,n}$ will be addressed later.
\vspace{-.1cm}
\begin{lem}
Based on the analysis in Lemma \ref{SSERT_34}, antenna gain $G(\theta_{m,n}^{i,j})$ can be approximated as
\vspace{-.25cm}
\begin{eqnarray}
G(\theta_{m,n}^{i,j})\leq \frac{\epsilon_a{(\lambda_{1})_{m,n}^{i,j}}}{{{\bf{w}}_{m,n}{\bf{A}}{\bf{w}}_{m,n}}} \nonumber\\
\leq \frac{\epsilon_a{(\lambda_{1})_{m,n}^{i,j}}}{2\mathcal{R}\Big\{ (({\bf{w}}_{m,n}^l)^H\tilde{\bf{v}})^H({\bf{w}}_{m,n}^H\tilde{\bf{v}})\Big\}-|({\bf{w}}_{m,n}^l)^H\tilde{\bf{v}}|^2}.\label{RSYSY}
\end{eqnarray}
\label{GSHTY}
\end{lem}
\vspace{-.5cm}
\begin{proof}
$(\lambda_{1})_{m,n}^{i,j}$ is the maximum eigenvalue for matrix ${\bf{v}}(\theta_{m,n}^{i,j}){\bf{v}}(\theta_{m,n}^{i,j})^H$. Recall that any convex function is globally lower-bounded by its first-order Taylor expansion and the denominator in (\ref{RSYSY}) at any point is derived as follows:
\vspace{-.3cm}
\begin{eqnarray}
{{\bf{w}}_{m,n}^H{\bf{A}}{\bf{w}}_{m,n}}\geq 
2\mathcal{R}\Big\{ (({\bf{w}}_{m,n}^l)^H\tilde{\bf{v}})^H({\bf{w}}_{m,n}^H\tilde{\bf{v}})\Big\}\nonumber\\ -|({\bf{w}}_{m,n}^l)^H\tilde{\bf{v}}|^2,
\label{YUSUI_89}
\end{eqnarray}
where ${\bf{w}}_{m,n}^l$ is one of the feasible points in the feasible region. It should be noted that Lemma \ref{GSHTY} guarantees the single-entry interference is strictly equal to or less than the antenna gain. We use $f_5({\bf{w}}^l_{m,n},{\bf{w}}_{m,n})$ to represent the rightmost term in Lemma \ref{GSHTY}.
\end{proof}
\vspace{-.2cm}

\subsubsection{Transformation for the Constraint (\ref{STEP2_con1})}
Constraints (\ref{STEP2_con1}) has a similar expression with (\ref{Prob3_con2_2}) and the similar transformations in (\ref{Prob3_con2_2}) thus can be adopted.

Based on Lemma \ref{GSHTY} and the transformation tricks in (\ref{Prob3_con2_2}), problem (\ref{Step_2}) can be transformed into:
\vspace{-.1cm}
\begin{subequations}
\label{P2_step2_final}
\begin{align}
\mathop {\max }\limits_{\bar{\bf{U}},{\bf{W}}}\quad & \sum_{m,n}\bar{u}_{m,n}{f}_2({\bf{W}}^{l},{\bf{W}})+\eta\sum_{m,n}f_3({\bar{u}}_{m,n}^l,{\bar{u}}_{m,n})
\label{sssr_45_1}\\
\text{s.t.}\quad 
&{f}_2({\bf{W}}^{l},{\bf{W}})\geq \bar{u}_{m,n}r_{m,n}\label{sssr_45_1_1}\\
&\sum_{m,n}\tilde{g}_{m,n}^{i,j}{f_5({\bf{w}}^l,{\bf{w}})}\bar{u}_{m,n}{\leq \sum_{m,n}\tilde{g}_{m,n}^{i,j}G(\theta_{m,n}^{\text{ low}})}\bar{u}_{m,n}
\label{sSETY_QIO},\\
&{{||{\bf{w}}_{m,n}||^2}{\leq p_{m,n}}}{}
\label{sSESCGD_4LO},\\
&{n\in \mathcal{N}_{m}^{\text{step 1}},~v_{i,j}\in \mathcal{V}^{\text{step 1}}}{}
\label{sSEIOWSME_4d},
\end{align}
\end{subequations}
where objective (\ref{sssr_45_1}) appends the penalty term $\eta\sum_{m,n}f_3({\bar{u}}_{m,n}^l,{\bar{u}}_{m,n})$ to fulfill the binary constraint for $\bar{u}_{m,n}$. To limit the interference, constraint (\ref{sSETY_QIO}) requires the actual antenna gain to be less than the modified bound $G(\theta_{m,n}^{\text{ low}})$. Note that ${f_5({\bf{w}}^l,{\bf{w}})}\bar{u}_{m,n}$ and $\bar{u}_{m,n}{f}_2({\bf{W}}^{l},{\bf{W}})$ are neither convex or concave with respect to ${\bf{w}}_{m,n}$ and $\bar{u}_{m,n}$. {To tackle the non-convexity, AO algorithm and the successive convex optimization techiniques can be applied in each iteration. Specifically, define $\{\bar{u}^r_{m,n}\}$ as the given SU node selection in the $r$-th iteration.
The beamforming $\{{\bf{w}}_{m,n}\}$ can be solved based on the SCA algorithm after fixing $\{\bar{u}^r_{m,n}\}$. Note that for $\bar{u}^r_{m,n}$ or ${\bf{w}}_{m,n}$, only the SU nodes belonging to ${\mathcal{U}^{\text{step 1}}}$ is considered. In a similar manner, $\{\bar{u}_{m,n}\}$ can be searched later after fixing $\{{\bf{w}}_{m,n}\}$. Due to page limitation, we omit the AO based algorithm for (\ref{P253_SDFSG}) and (\ref{P2_step2_final}) here.}

%% file: Numerical_results.tex
\vspace{-.2cm}
\section{Numerical Evaluation}
\label{POIUIST}
\vspace{-.1cm}

In this section, we provide numerical examples to demonstrate the effectiveness of the proposed algorithms. The CVX solver is applied to solve the disciplined convex programming problems. Unless otherwise specified, the same system settings and SU network deployment are used in the simulations. 
As defined in \cite{bahrak2014protecting}, we define the database coverage area regulated by the SAS, as a 300km by 300km square. In such a database coverage area, the detailed definition of EZ, SUs and PUs is given below: (a) we consider an EZ in this region, which is established to reduce the interference to PUs. For a typical setting of EZ, we set the radius of the EZ to be equal to 20km, which will be varied later under different settings; (b) SU nodes, including STs and SRs, are scattered along the edge of the EZ. We consider 12 STs deploying along the edge of the EZ, where each ST serving 9 SRs in each cell and 108 SRs in total are considered. For each ST, the coverage radius is set to be 150m. For each SR node, noise variance is set to be $\sigma_n^2=-94$ dBm. For the downlink transmission of STs, the 3GPP urban pathloss model with a path exponent of 3.6 is considered. (c) In this work, we consider a PU, which stays within the privacy zone. The large-scale propagation loss model for interference is defined in \cite{sss2013NIIT}. The interference from SU nodes contains many terms including the transmitter/receiver insertion loss (2 dB), and cable loss. More details can be found in \cite{sss2013NIIT}. Unless otherwise specified, we consider 121 cells in the privacy zone candidate set $\mathcal{C}_c$. {There are many ways that the number of cells in the privacy zone will affect the design. For instance, configuring each with the large size can cause PUs to experience different large scale fadings in the same cell. At the same time, challenges such as large-scale mixed integer programming and excluding more SUs cannot be neglected. We leave such problems in our future research.} {Please note that we only conduct numerical simulations in this work and the experimental analysis on the real-life datasets will be presented in the future based on the measurements from commercial 5G wireless networks.}     

Problems involving integer programming are difficulty to converge due to large variation of performance. Before presenting the numerical results of the proposed scheme, we provide simulation results to illustrate the convergence performance and specify the proper iteration numbers. {In Fig. \ref{Fig.sub.243_CR1}, we present the convergence performance for the SCA algorithm-based SU data rate maximization. The x-axis represents the iteration numbers and the y-axis represents the difference between current and last objective values. 
From Fig. \ref{Fig.sub.243_CR1}, under various aggregate interference thresholds, we observe that the SCA algorithm usually converges after 2 or 3 iterations. Besides, when the specified threshold $I_{th}$ is relaxed from -116dBm to -112dBm, the algorithm quickly converges in about 3 steps.} This is due to the fact that a lower interference threshold induces less fluctuation over the convergence threshold. Based on the analysis above, we set the iteration number under different interference thresholds to be 4 in the following simulations. 
\vspace{-.35cm}
\begin{figure}
   \centering
    \includegraphics[width=0.84\linewidth,height=0.64\linewidth]{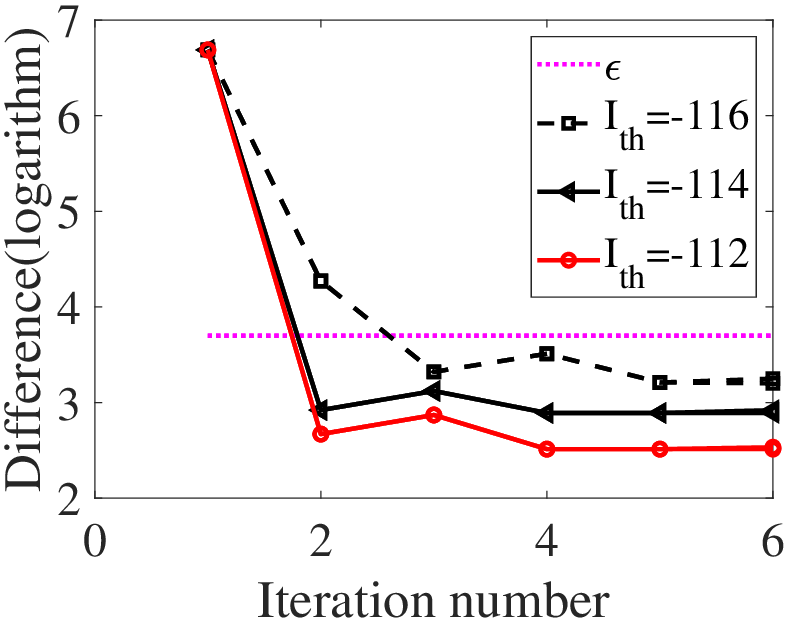}
    \caption{Convergence vs. Iteration Number.}
    \label{Fig.sub.243_CR1}
\end{figure}

\begin{figure}[htbp]
    \centering
    \subfigure[The sum rate under various privacy zone size.]{ 
    \includegraphics[width=0.84\linewidth,height=0.7\linewidth]{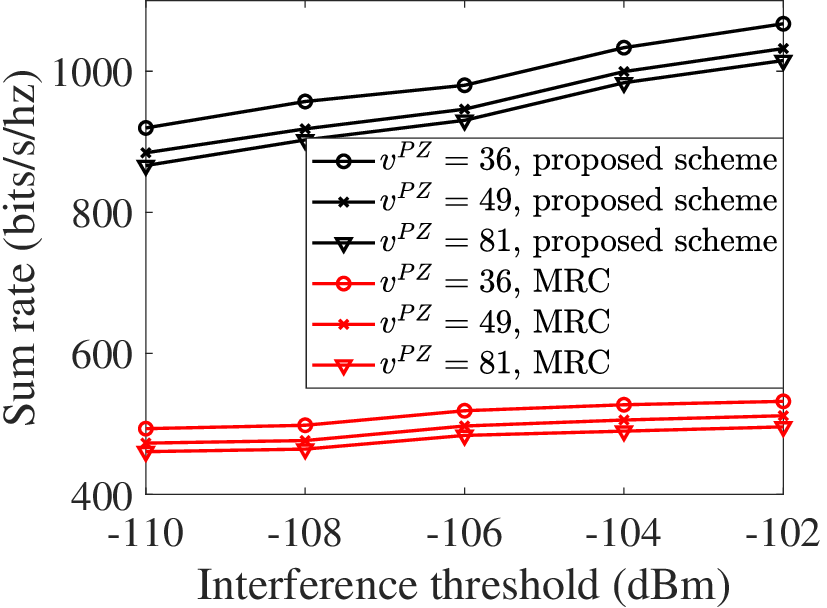}
    \label{fig:CR_pic9}
    }
    \subfigure[SU exclusion, $N_{re}=10$]{ 
    \includegraphics[width=0.8\linewidth,height=0.68\linewidth]{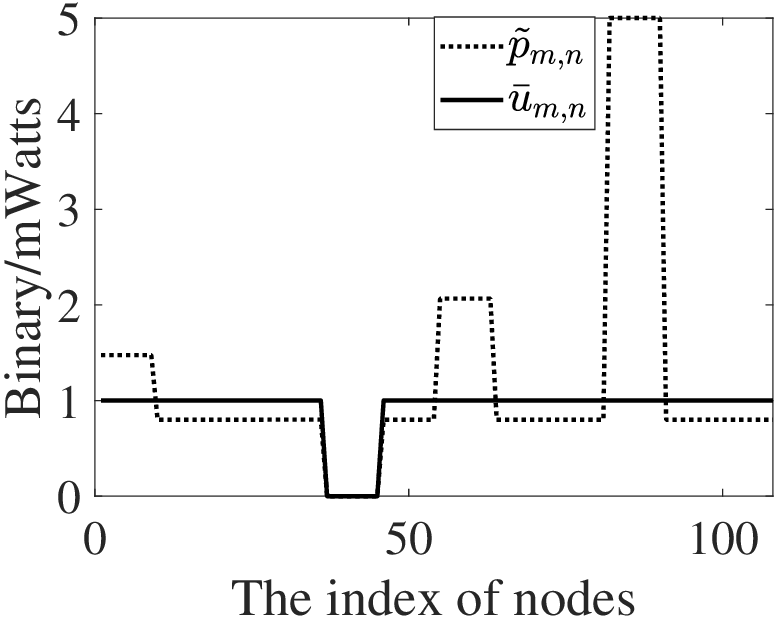}
    \label{Fig.sub.123_CR89}
    }
    \caption{The Sum Rate and SU Node Exclusion}
    \vspace*{-.1cm}
\end{figure}

\vspace{-.2cm}
\subsection{The Performance of SU Data Rate}

We first evaluate the performance of SU data rate maximization design given the PU location privacy constraints, i.e., with the established privacy zone. It is not difficult to find that the SU data rate of existing schemes can be significantly downgraded in an effort to reduce the aggregate interference below the interference threshold of the privacy zone. In Fig. \ref{fig:CR_pic9}, we compare the sum rates obtained by the proposed Algorithm 1 and the MRC-based scheme \cite{li2015performance} under various privacy zone size {$v^{\text{PZ}}$}. When $v^{\text{PZ}}$ becomes sufficiently large after expanding the privacy zone, e.g., $v^{\text{PZ}}$ = 81, the signal travel distance between STs and PUs is greatly reduced and the low transmission power is desired to avoid severe interference to the PU, which can result in the low sum rate. From Fig. \ref{fig:CR_pic9}, we can easily observe this trend, i.e., increasing {$v^{\text{PZ}}$} leads to the downgraded sum rate of the MRC-based scheme. The sum rate of the proposed scheme, however, is less affected by the large {$v^{\text{PZ}}$}, i.e., for a given interference threshold -114dBm, the proposed scheme can achieve the sum rate as high as 854 bits/s/Hz. This is because the proposed scheme is able to efficiently mitigate the high interference caused by the short travel distance via beamforming. For instance, the proposed scheme with beamforming enables each ST to focus its mainlobe/sidelobe over the SRs, while putting the deep null over the angular 
in the direction of the privacy zone on the beampattern.

To prove the effectiveness of $\tilde{p}_{m,n}$ in (P1.1), formulated according to the Big-M formulation, one realization of the optimal SU node exclusion is given in Fig. \ref{Fig.sub.123_CR89}. With the SU exclusion factor $N_{re}=10$, nearly 10 SUs are excluded from the current channel in this realization. In this figure, for the ease of comparison, we normalize the level of allocated power between $[1,5]$, which is represented in the dotted line. The binary indicator $\bar{u}_{m,n}$ is denoted by the line in dark. It is observed that different power levels are assigned to SUs based on the channel conditions. Moreover, to maximize the power utilization efficiency, the power of the excluded nodes, with $\bar{u}_{m,n}=0$, is also forced to be 0. This is attributed to the fact that given the SU exclusion factor $N_{re}$, the proposed scheme will begin by excluding the SU nodes causing the large interference to PUs, which is beneficial to incorporate more SU nodes on the band. Then, for the SU nodes with weak channel conditions, more power will be assigned accordingly, which aims to meet the minimum QoS constraints.

\vspace{-.2cm}
\subsection{The Performance of PU Location Privacy}

\paragraph{PU Location Privacy V.S. Various SU Exclusion Factors}In Fig. \ref{fig:CR_pic1}, we compare the PU location privacy (in bits) achieved by the following schemes: 1) The proposed scheme with SU node exclusion, i.e.,  $N_{re}=30$ or $N_{re}=45$, which is obtained by solving problem (\ref{P253_SDFSG}); 2) The proposed scheme without SU node exclusion, which is obtained by solving problem (\ref{P253_SDFSG}) after setting $N_{re}=0$; and 3) the MRC beamformer-based scheme, where the interference control is not considered during the beamformer design. It is chosen as the baseline. There are 120 SR nodes considered in the simulation. For every single-antenna SR, the range of power allocation is set to be between 26 dBm and 50 dBm. One can observe from Fig. \ref{fig:CR_pic1} that when the interference threshold is below than -104 dBm, the proposed scheme has a higher PU location privacy compared with other schemes. In particular, the proposed scheme with $N_{re}=30$ or $N_{re}=45$, can operate under the low thresholds, i.e., -106dBm and achieve the location privacy 6.91. In contrast, the MRC-based scheme cannot find the feasible privacy zone if the threshold $I_{th}<-104~ \text{dBm}$ because its location privacy quickly drops to zero after $I_{th}<-104 ~\text{dBm}$. This is expected since with the interference mitigation capability enabled by beamforming, for the proposed scheme, the transmitted signals within the angular space, that is adjacent to the privacy zone, are highly attenuated by forcing its antenna gains to approach to zero, which can easily cause a 30dB to 40dB attenuation to the interference from STs. Higher interference attenuation in the direction of privacy zone leads to a larger expansion margin for the privacy zone, which thus can maintain the PU location privacy if $I_{th}<-104~ \text{dBm}$. In addition, when PUs have receivers with the high sensitivity requirements on the unintentional interfering signals, the low threshold, i.e., $I_{th}<-104~\text{dBm}$, is usually demanded. 
The MRC-based scheme cannot find a feasible privacy zone in this case, which forces the SAS to exclude all the SU nodes on the shared band to decrease the interference level at PUs and prioritize its availability. If this occurs, the PU location privacy is assigned to be 0 in the figure for the ease of presentation. The interference threshold triggering the spectrum sharing service outage is called cutoff threshold. The proposed scheme with the SU nodes exclusion has the lower cutoff thresholds. {When $N_{re}=30$ or $N_{re}=45$, the cutoff threshold of the proposed scheme are -108dBm and -110 dBm while that of MRC-based scheme is -104 dBm. The -4 dBm difference in the threshold has a huge impact on the PU location privacy protection.} In Fig. \ref{Fig.main_CR1STS}, given a squared candidate set with $|\mathcal{C}_c|=121$ and $N_{re}=10$, two privacy zones, 
represented by the cells in yellow, are derived. If the lower interference threshold is required, i.e., reducing the threshold from -106 dBm to -108 dBm, it is observed that the privacy zone can shrink significantly. The proposed scheme achieves lower interference thresholds, i.e 2dBm to 4dBm lower than that of MRC-based scheme. For the airborne or shipborne radars that operates in environments 
requiring very low interference, the proposed scheme is capable to satisfy the requirement in this case.
\vspace{-.35cm}
\begin{figure}[htbp]
    \centering
    \subfigure[Location privacy versus interference threshold under various SU exclusion factor.]{ 
    \includegraphics[width=0.84\linewidth,height=0.72\linewidth]{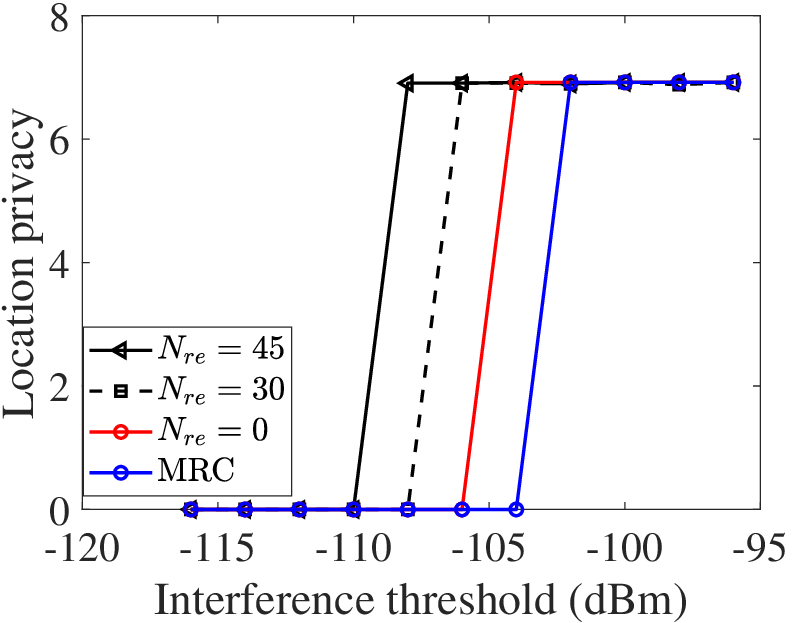}
    \label{fig:CR_pic1}
    }
    \subfigure[Location privacy versus the number of nodes under various SU exclusion factor]{ 
    \includegraphics[width=0.84\linewidth,height=0.72\linewidth]{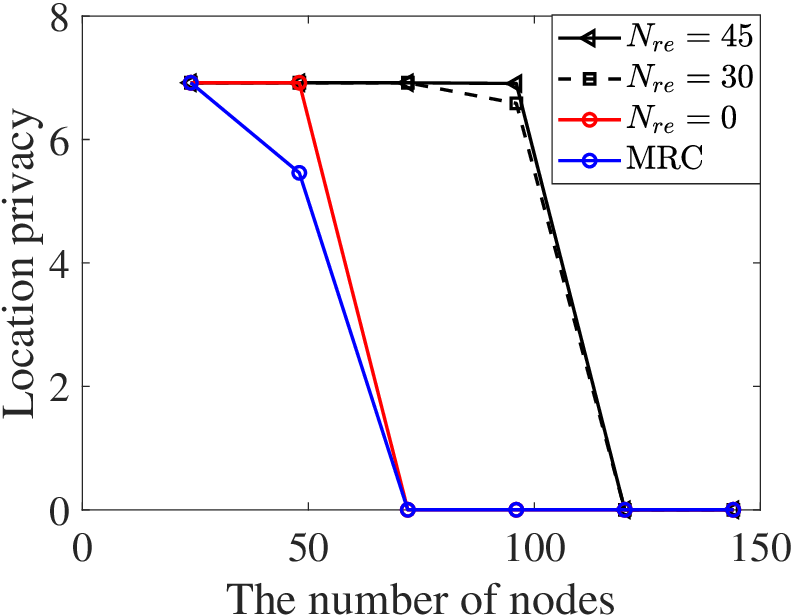}
    \label{fig:CR_pic2}
    }
    \caption{Location Privacy Under Various SU Node Exclusion Factors}
    \vspace*{-0.2cm}
\end{figure}

\paragraph{PU Location Privacy V.S. the Number of SU Nodes}We compare the location privacy (in bits) with various numbers of the SU nodes in Fig. \ref{fig:CR_pic2} with the interference threshold $I_{th}=-108 ~\text{dBm}$ and the EZ radius $EZ=20~\text{km}$. Several important observations can be found from Fig. \ref{fig:CR_pic2}. First, as expected, the PU location privacy deceases as the number of SU nodes increases. This is because serving more SU nodes directly increases the level of aggregate interference in the privacy zone, especially for the cells at the edge of privacy zone, which makes the marginal area in the privacy zone to gradually become infeasible for the interference threshold and finally lead to the shrunk privacy zone. Second, by comparing the proposed scheme and MRC-based scheme, the proposed scheme has obvious advantages while serving a large number of SU nodes. When massive SU nodes, i.e., 100 SU nodes, are considered, the proposed scheme with the SU exclusion factor $N_{re}=45$ or $N_{re}=30$ can maintain a high PU location privacy equal to 6.91, whereas the MRC-based scheme cannot establish a feasible privacy zone with more than 50 SU nodes. This observation indicates that the proposed scheme is capable to manage the increased aggregate interference caused by the massive SU transmissions. By using various SU exclusion factors, the proposed scheme can exclude the SU nodes that are causing severe interference with the antenna mainlobe pointing to the privacy zone, and keep the SU nodes with links inducing less interference. For this reason, the proposed scheme only needs to exclude a small number of SU nodes to maintain the location privacy level. In consequence, the proposed scheme has the resilience to handle a surge in the spectrum sharing requests with the massive SU nodes.

\vspace{-.35cm}
\begin{figure}[]
\centering  
\subfigure[$N_{re}=10$, $|\mathcal{C}_c|=121$, $I_{th}=-108$ dBm]{
\label{Fig.sub.1_CR3}
\includegraphics[width=0.86\linewidth,height=0.6\linewidth]{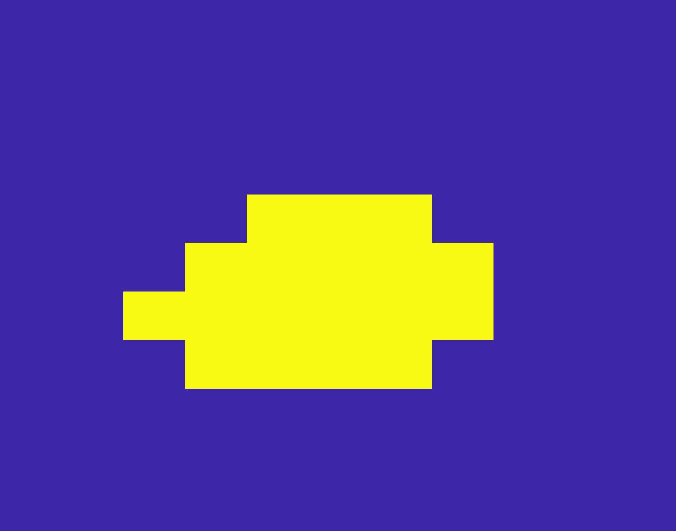}}
\subfigure[$N_{re}=10$, $|\mathcal{C}_c|=121$, $I_{th}=-106$ dBm]{
\label{Fig.sub.2_CR1}
\includegraphics[width=0.86\linewidth,height=0.6\linewidth]{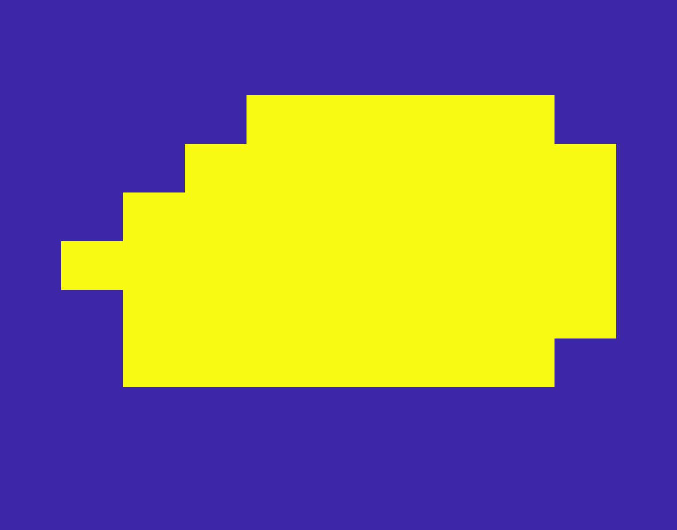}}
\caption{Privacy Zones Design Under Various Interference Threshold}
\label{Fig.main_CR1STS}
\vspace*{-.2cm}
\end{figure}

\vspace{.3cm}
\paragraph{PU Location Privacy V.S. the EZ Radius}In addition to the number of SU nodes, the EZ radius also has a direct impact on the performance of the PU location privacy. An intuitive benefit coming after increasing the EZ radius is that the interference from ST nodes can be reduced clearly, which results in an expanded privacy zone. However, the negative effect of a large EZ radius nor can be ignored. After enlarging the EZ radius, all of SU nodes within the EZ has to be excluded from the current band correspondingly, which leads to the low spectrum/space utilization rate. Therefore, a scheme, providing the high PU location privacy under the small EZ radius, is highly desired. In Fig. \ref{fig:CR_pic4} and Fig. \ref{fig:CR_pic3}, the location privacy under the different EZ radii is given. The location privacy of the proposed scheme with exclusion factor $N_{re}=20$ (denoted as Ex), the proposed scheme with $N_{re}=0$, and the MRC-based scheme are represented by the black lines, red lines and blue lines, respectively. In Fig. \ref{fig:CR_pic4}, it is observed that for the relatively small EZ radius, the location privacy of the proposed scheme with $N_{re}=20$ begins to decrease at $I_{th}=-105$ dBm, whereas the location privacy of MRC-based scheme quickly approaches 0 when $I_{th}=-102$ dBm with $\text{EZ}=18~\text{km}$ or $I_{th}=-101$ dBm with $\text{EZ}=20~\text{km}$. Such a performance gain is mainly attributed to the fact that the MRC-based scheme cannot mitigate the strong interference associated with the short interference travel distance after decreasing the EZ radius and thus demands a high interference threshold to search for a feasible privacy zone. In contrast to the MRC-base scheme, by adjusting the beamformer, the proposed scheme can avoid the strong interference from ST nodes pointing to the marginal area of the privacy zone, which leaves more flexibility in reponse to the increased interference due to  reducing the EZ. The same phenomenon can be observed in Fig. \ref{fig:CR_pic3} while investigating the impact of the small EZ radius over the SU nodes numbers. With the interference threshold $I_{th}=-104~\text{dB}$ and user exclusion factor $N_{re}=20$, for the MRC-based scheme, its location privacy starts to drop quickly while serving more than 50 SU nodes with  $\text{EZ}=18\text{km}$. In contrast to the MRC-based scheme, the location privacy of the proposed scheme $N_{re}=0$ begins to decrease while serving over 70 SU nodes on the current band with $EZ=18$km. The proposed scheme with the SU node exclusion still can achieve high location privacy while serving 100 SU nodes. In summary, the above observations demonstrate the proposed scheme is capable to maintain the high PU location privacy and support more SU nodes under the small EZ radius case, which thus is desired in the dense networks. 

\begin{figure}[htbp]
    \centering
    \subfigure[Location privacy versus interference threshold under various radius of EZ.]{ 
    \includegraphics[width=0.84\linewidth,height=0.7\linewidth]{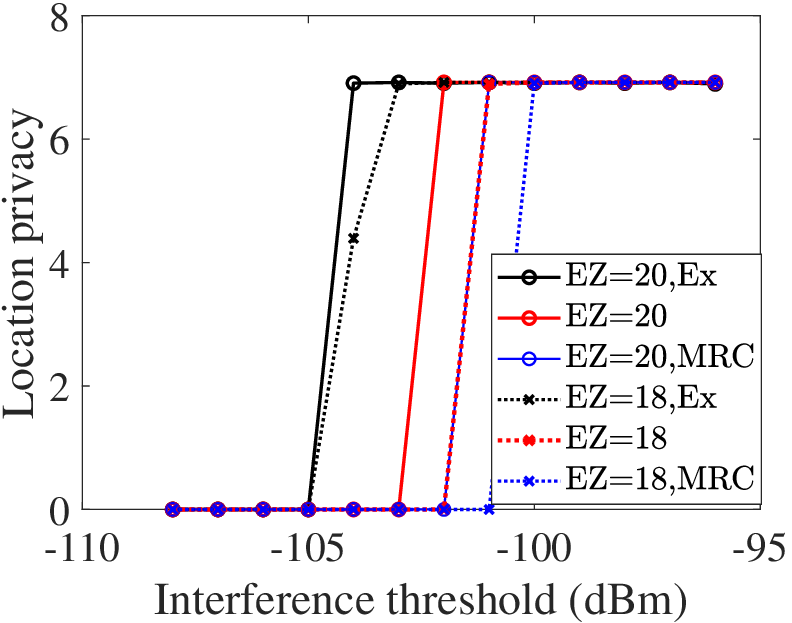}
    \label{fig:CR_pic4}
    }
    \subfigure[Location privacy versus the number of nodes under various radius of EZ.]{ 
    \includegraphics[width=0.84\linewidth,height=0.7\linewidth]{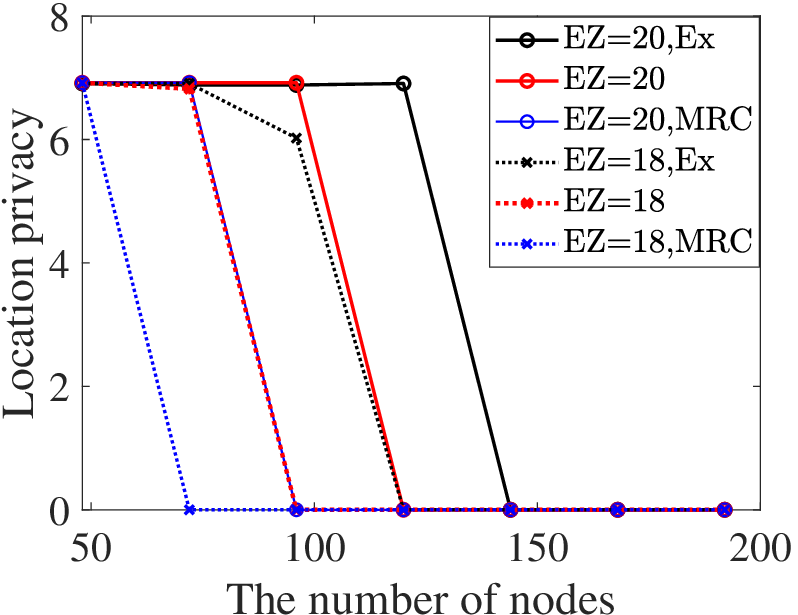}
    \label{fig:CR_pic3}
    }
    \caption{Location Privacy Under Various EZ Radius}
    \vspace*{-.2cm}
\end{figure}

\vspace{-.3cm}
\begin{figure}[htbp]
    \centering
    \subfigure[Location privacy versus interference threshold under various antenna gain ratio.]{ 
    \includegraphics[width=0.84\linewidth,height=0.7\linewidth]{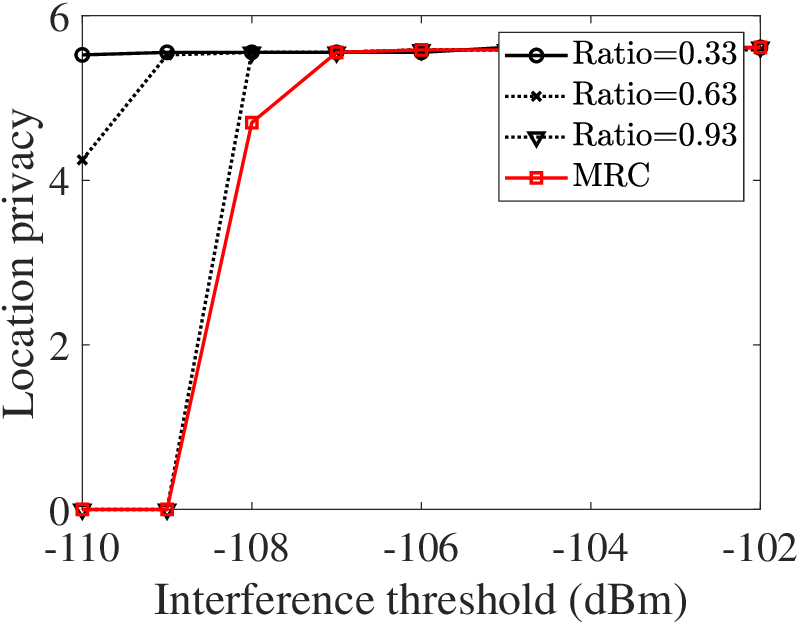}
    \label{fig:CR_pic5}
    }
    \subfigure[The number of nodes versus interference threshold under various antenna gain ratio.]{ 
    \includegraphics[width=0.84\linewidth,height=0.7\linewidth]{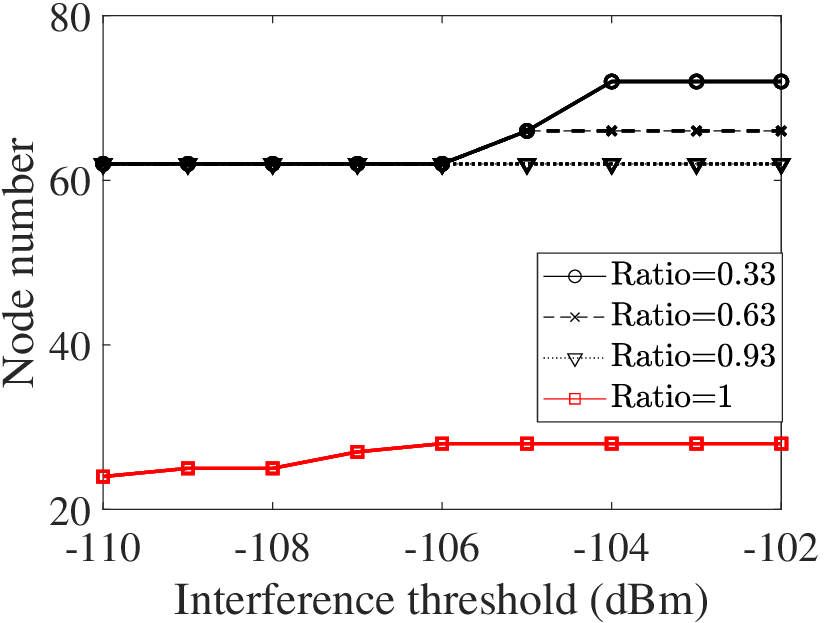}
    \label{fig:CR_pic6}
    }
    \caption{Location Privacy Under Antenna Gain Ratios}
\vspace{-.1cm}
\end{figure}

\vspace{.2cm}
\paragraph{PU Location Privacy V.S. the Antenna Gain Ratios}With beamforming ability, the proposed scheme can restrict its antenna gain and operate with the low antenna gain specified by the antenna gain ratio. The reduced antenna gain leads to a decreased upper bound for the interference. The antenna gain ratio thus can directly affect the PU location privacy. In Fig. \ref{fig:CR_pic5} and Fig. \ref{fig:CR_pic6}, we investigate the impact of the various antenna gain ratio on the PU location privacy and SU node numbers operating on current networks. We consider the PU location privacy achieved by three designed schemes, i.e., 1) the proposed scheme with SU node exclusion $N_{re}=10$; 2) the proposed scheme without SU node exclusion; 3) the MRC-based scheme. The privacy zone candidate set contains 49 candidate cells and the spectrum sharing system aims to serve 78 SU nodes. From Fig. \ref{fig:CR_pic5}, it is observed that the proposed scheme under the low antenna gain ratios such as 0.33 or 0.63 achieves higher PU location privacy than that of the MRC-based scheme at various thresholds. Owing to the beamforming, the proposed scheme can achieve a low antenna gain after specifying a low antenna gain ratio. As a result, the upper bound on the interference constraint is further decreased. The reduced interference upper bound will increase the margin between the interference threshold $I_{th}$ and the aggregate interference. In this case, the proposed scheme can achieve more flexibility for the interfere control while expanding the privacy zone and allocating the transmission power. The higher transmission power and more SU nodes will be tolerated. For instance, under the low antenna gain ratios, the proposed scheme can maintain the high PU location privacy with $I_{th}=-109$ dBm, whereas the MRC-based scheme requires a minimum interference threshold $I_{th}=-107$ dBm to keep the same location privacy level. This phenomenon also can be reflected in the number of SU nodes served on the current band. In Fig. \ref{fig:CR_pic6}, it is obvious that the propose scheme can sustain more SU nodes on the current band via using the low antenna gain ratios. This is because the reduced aggregate interference due to the low antenna gain ratio allows more SU nodes to access the current band with the constrained transmission power level. In contrast, the MRC-based scheme has no control over the antenna gain according to its design principle and thus cannot restrict its antenna gain. Consequently, the MRC-based scheme cannot be reactive to a serge in the SU nodes and cannot serve the massive SU nodes.

\begin{figure}[htbp]
    \centering
    \subfigure[The summation rate under various EZ radius.]{ 
    \includegraphics[width=0.84\linewidth,height=0.7\linewidth]{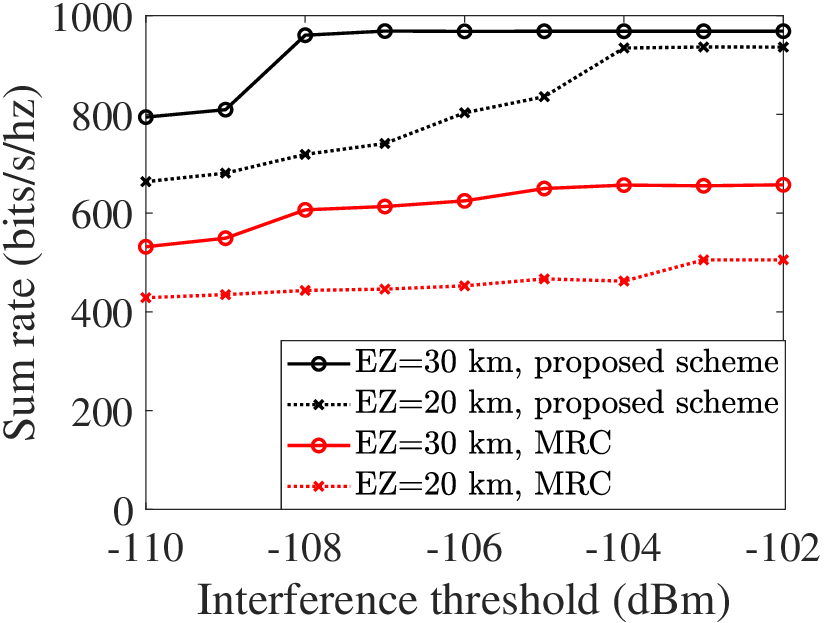}
    \label{fig:CR_pic7}
    }
    \subfigure[The summation rate under various antenna gain ratio.]{ 
    \includegraphics[width=0.84\linewidth,height=0.7\linewidth]{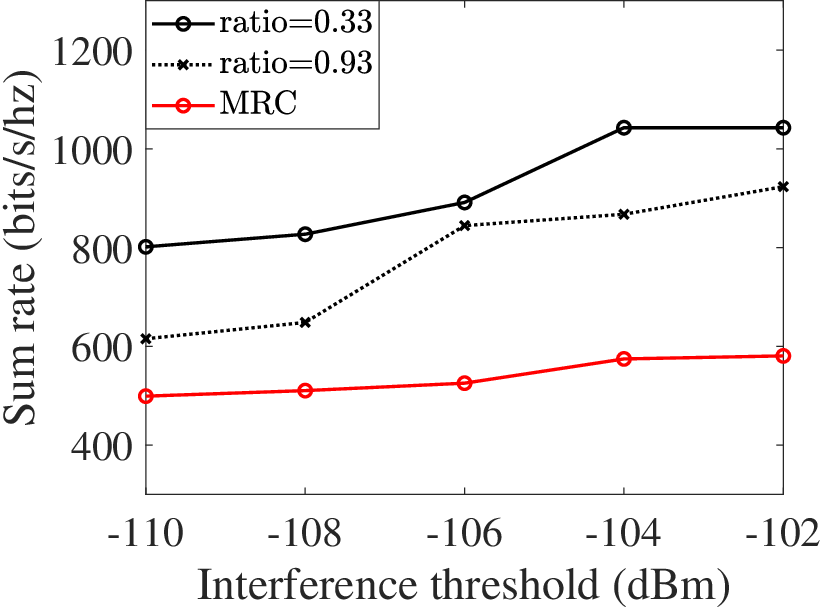}
    \label{fig:CR_pic8}
    }
    \caption{Sum Data Rate Under Various EZ Radius and Antenna Gain Ratios}
    \label{fig:8_2}
    \vspace*{-.1cm}
\end{figure}

\vspace{-.2cm}
\subsection{The Performance of SU Data Rate for Problem (\ref{P2_step2_final})}

In Fig. \ref{fig:8_2}, we compare the summation rate obtained by (\ref{P2_step2_final}) under various EZ radii and antenna gain ratios. Note the same number of SU nodes is considered throughout this analysis. It is observed from Fig. \ref{fig:CR_pic7} that for the proposed scheme and MRC-based scheme, the summation rate tends to become saturated when $I_{th}$ is sufficiently large. In addition, the summation rate of the proposed scheme converges more rapidly after the EZ radius increases from 20km to 30km. In contrast, the convergence pace of the MRC-based scheme is still very slow even the EZ radius increase. Enlarging the EZ requires the transmitted signals to travel a longer distance, which forces the interference to have a larger attenuation and thus helps improve the summation rate. The proposed scheme takes advantage of this benefit to achieve a better trade-off between increasing transmission power and reducing the aggregate interference by using beamforming. The MRC-based scheme cannot be adaptive to the enlarged EZ radius in time, which results in a low convergence pace. In Fig. \ref{fig:CR_pic8}, we also present the summation rate under various antenna gain ratios. As the decreasing of antenna gain ratios, the proposed scheme can restrict the antenna gain that corresponds to the privacy zone by performing beamforming in the angular space, which helps decrease the upper bound for the aggregate interference and the transmitted signal can be assigned with the higher transmission power without inducing large interference to the privacy zone. As a result, the high summation rates and desirable PU location privacy can be achieved at the same time.

\vspace{-.2cm}
\section{Conclusion}
\label{Con_cl}

In this paper, we investigate the benefits of beamforming techniques on the PU location privacy enhancement. Two scenarios are considered. In the first problem, we aim to improve the SU network communication performance with the specified PU location privacy requirements. In the second problem, we intend to improve PU location privacy given constraints on the QoS requirement of SUs. Extensive numerical evaluations are conducted. The numerical results show that in the first scenario, the proposed scheme can serve more SU nodes with higher communication throughput and also satisfy the specified PU location privacy requirements. In the second scenario, the proposed scheme enables to configure a much larger privacy zone while satisfying the SU network throughput requirement.

%% file: Appendix2.tex
\vspace{-.2cm}
{\centering{\section*{Appendix I}}}

In (\ref{f1_interference}), $G(\theta_{m,n}^{i,j})$ is the antenna gain for $u_{m,n}$ and it is given as:
\vspace{-.25cm}
\begin{eqnarray}
G(\theta_{m,n}^{i,j})=\frac{\epsilon _{a}~{{\max_k}}\big(U(\theta_k)\big)}{\frac{1}{4\pi}\int_{0}^{2\pi}\int_{0}^{\pi}U(\theta)\sin{\theta}\text{d}\theta\text{d}\phi}.
\label{IOPUSY}
\end{eqnarray}
The antenna gain is expressed in terms of the radiation intensity $U(\theta_k)$ towards the angle of arrival $\theta_k$. $\epsilon _{a}$ is the antenna efficiency and it is a constant. The definition of the radiation intensity is given as $U(\theta_k)={\bf{w}}_{m,n}^H{\bf{v}}(\theta_{k}){\bf{v}}(\theta_{k})^H{\bf{w}}_{m,n}.$ In the proposed scheme, antenna gain and directivity are in comparison to an isotropic antenna. In general, for the isotropic antenna element, the radiation intensity involving the array factor is $U(\theta_{k})
={{\bf{w}}_{m,n}^H{\bf{v}}(\theta_{k}){\bf{v}}(\theta_{k})^H{\bf{w}}_{m,n}}$. Here ${\bf{w}}_i$ is the beamforming vector and ${\bf{v}}(\theta_k)$ is the antenna steering vector on the quantized angle of arrival $\theta_k=x_{\Delta}*k,~k\in\{1,...,K\}$. ${x_{\Delta}}$ is the quantization level. For every single beamformer, the average radiation intensity is given as $\overline{U}_o^{m,n}=\frac{1}{4\pi}\int_{0}^{2\pi}\int_{0}^{\pi}U(\theta)\sin{\theta}\text{d}\theta\text{d}\phi
$. Based on the Trapezoid rule for the integral, the average antenna gain $\overline{U}_o^{m,n}$ is approximated as:
\vspace{-0.15cm}
\begin{eqnarray}
\overline{U}_o^{m,n}&\approx &
{x_{\Delta}}\Big(\sum_{n\in I\setminus I_1}\frac{U(\theta_n)\sin(\theta_n)}{2}+
\sum_{n\in I_1}{U(\theta_n)\sin(\theta_n})\Big)\nonumber
\label{RTSY_89SUS}
\end{eqnarray}
where the index set is defined as $I=\{I_1,1,K\}$. In (\ref{RTSY_89SUS}), the approximation for $\overline{U}_o^{m,n}$ can be represented into $\overline{U}_o^{m,n}\approx{\bf{w}}_{m,n}^H{\bf{A}}{\bf{w}}_{m,n}$, where the matrix ${\bf{A}}$ is defined as below:
\vspace{-.25cm}
\begin{eqnarray}
{\bf{A}}=\sum_{n\in  I_1}{x_{\Delta}\sin(\theta_n)}{\bf{v}}(\theta_n){\bf{v}}(\theta_n)^H \nonumber\\
+\sum_{n\in I\setminus I_1}\frac{x_{\Delta}\sin(\theta_n)}{2}{\bf{v}}(\theta_n){\bf{v}}(\theta_n)^H
\label{SVCTR_34}
\end{eqnarray}

For a given AoA $\theta_k$, the ratio of its maximum radiation intensity $U(\theta_k)$ to its mean radiation intensity is the antenna directivity $D(\theta_k)$. The power gain or simply gain $G(\theta_n)$ (in watts) of an antenna in a given direction takes efficiency into account and is defined as the product of the efficiency factor $\epsilon_{a}$ and the antenna directivity. It is given as $G(\theta_{m,n}^{i,j})=\epsilon _{a}\cdot D(\theta_{m,n}^{i,j})=\epsilon _{a}\cdot\frac{\max\big(U(\theta_k)\big)}{\overline{U}_o^{m,n}}$.

\vspace{-.2cm}
{\centering{\section*{Appendix II}}}

To find the upper bound for the antenna gain $G(\theta_{m,n}^{i,j})$, we need to approximate its numerator and denominator, respectively. For antenna gain defined in Appendix I, please note we can easily prove that ${\bf{A}}$ is a positive semi-definite (PSD) matrix. For a PSD matrix, matrix ${\bf{A}}$ has the eigendecomposition as ${\bf{A}}={\bf{Q}}{\bf{\Gamma}}{\bf{Q}}^H$. We can find the following inequality for the average radiation intensity $\overline{U}_o^{m,n}\approx {\bf{w}}_{m,n}^H{\bf{A}}{\bf{w}}_{m,n}$ as:
\vspace{-.25cm}
\begin{eqnarray}
{\bf{w}}_{m,n}^H{\bf{A}}{\bf{w}}_{m,n}=\sum_{i=1}^{N}\lambda_i ||{\bf{w}}_{m,n}^H{\bf{q}}_i||^2_2\geq(\lambda_{min})_{A}||{\bf{w}}_{m,n}||^2_2.
\label{WERT_78ER}
\end{eqnarray}

Similarly, we can find the inequality for the radiation intensity in the numerator as:
\vspace{-.25cm}
\begin{eqnarray}
\max\big(U(\theta_{k})\big)=\sum_{i=1}^{N}\tilde{\lambda}_i ||{\bf{w}}_{m,n}^H{\bf{p}}_i||^2_2 \leq \tilde{\lambda}^1_{m,n}||{\bf{w}}_{m,n}||^2_2.
\label{QWERT_rt5}
\end{eqnarray}

This is due to the fact for any direction $k$ with maximum antenna factor ${\bf{w}}_{m,n}^H{\bf{v}}(\theta_k){\bf{v}}(\theta_k)^H{\bf{w}}_{m,n}$, the eigen-value $\tilde{\lambda}^1_{m,n}$ of matrix ${\bf{B}}={\bf{v}}(\theta_k){\bf{v}}(\theta_k)^H$ is equal to a constant. In this case, the matrix ${\bf{B}}$, containing the replica vector, can be viewed as the base of DFT matrix. By combining the result of (\ref{WERT_78ER}) and (\ref{QWERT_rt5}), we can get the upper bound of antenna gain for instant beamformers.